\documentclass[a4paper,onecolumn, unpublished]{quantumarticle}
\pdfoutput=1
\usepackage{note}
\usepackage[utf8]{inputenc}
\usepackage[english]{babel}
\usepackage[T1]{fontenc}
\usepackage{pifont}

\usepackage{minitoc}
\usepackage{mathtools}
\usepackage{microtype}

\usepackage[backend=bibtex, style=numeric,citestyle=numeric-comp, sorting=none,sortcites=true]{biblatex}
\usepackage{stmaryrd}
\usepackage{CJKutf8}

\begin{document}
\title{Adaptive Loss-tolerant Syndrome Measurements}
\author{Yuanjia Wang}
\email{raycosine@gmail.com}
\orcid{0000-0002-8147-2370}
\affiliation{Department of Electrical \& Computer Engineering, University of Southern California, Los Angeles, California}

\author{Todd A. Brun}
\affiliation{Department of Electrical \& Computer Engineering, University of Southern California, Los Angeles, California}

\date{March 18, 2026}

\maketitle
\begin{abstract}
In the presence of qubit losses, the building blocks of fault-tolerant error correction (FTEC) must be revisited. Existing loss-tolerant approaches are mainly architecture-specific, and little attention has been given to optimizing the syndrome measurement sequences under loss. Schemes designed for the standard Pauli error model are not directly applicable because the syndrome patterns differ when both Pauli errors and erasures can occur.

Based on recent advances in loss detection units and loss-tolerant syndrome extraction gadgets, 
we extend the study of adaptive Shor-style measurement sequences to the mixed error model. 
We begin by discussing how to adaptively convert correctable erasures into located errors. 
The minimal overhead is quantified by the number of stabilizer measurements, which can be reduced to a subgroup dimension problem for erasures arising in any FTEC circuit for qubits and prime-dimensional qudits. As a byproduct, we provide the construction of the canonical generating set with respect to a given bipartite partition for a stabilizer group on qudits of composite dimension. We then generalize both the weak and strong FTEC conditions. Finally, we present adaptive syndrome-measurement protocols for the mixed error model, generalizing the adaptive protocols for the standard Pauli error model. 
\end{abstract}

\section{Introduction}\label{sec:Introduction}

Quantum error-correcting codes (QECCs) protect quantum information from errors. In certain scenarios, such as communication, quantum states are transmitted through noisy channels, while the encoding and decoding operations are assumed to be perfect. However, quantum error correction must be implemented fault-tolerantly in general quantum-computation scenarios, since noise can occur at any point during the computation, including encoding and decoding.

Fault tolerance (FT) is guaranteed by performing \emph{syndrome extraction} (SE), which heavily relies on measurements \cite{bartolucci2023fusion}. For stabilizer codes, it suffices to extract a usable string containing the full syndrome information of the stabilizer group from 
noisy measurements, and then use it to identify and correct physical errors with a classical decoder. However, measurements on many current architectures are much slower than gate operations \cite{acharya2022suppressingquantumerrorsscaling, monroe2021faulttolerantcontrol, foxen2020continuous, gottesman2022opportunities}, hence introducing substantial time overhead.

Prior work has explored various approaches to reducing this overhead. First, for a large family of codes, syndrome extraction can be performed less frequently when fault tolerance is required only at the level of the entire algorithm, rather than for every single gadget \cite{zhou2025low}. The two components of syndrome extraction can also be optimized, respectively. In the first stage, \emph{syndrome measurement}, the goal is to reduce the number of measurements as much as possible, which is the primary focus of this work. The second stage, \emph{classical decoding}, requires efficient and scalable classical algorithms and is not the central focus of this work.

Efforts in these three directions are mutually compatible, allowing us to narrow the focus on reducing the number of measurements. The eponymous weight reduction techniques \cite{hastings2016weight, hastings2021quantum,  hsieh2025simplified, sabo2024weight} transform any quantum LDPC code to a new code with constant-weight stabilizers while maintaining good rate and distance. Another viable approach is to reduce the number of stabilizers that need to be measured. 
It is then worth addressing this question:

\emph{What is the minimum number of stabilizer measurements required to extract a usable syndrome string?
}

We will also primarily focus on the Shor-style syndrome extraction, which sequentially measures syndromes using a fresh ancillary cat state in every stabilizer measurement. The traditional Shor-style SE is costly and requires $O(d^2)$ rounds of syndrome string extraction in the worst case. There has recently been a series of works focusing on 
reducing the length of syndrome extraction sequences: \cite{delfosse2020short,delfosse2021beyond} on single-shot and sub-single-shot schemes satisfying weak FT conditions, \cite{tansuwannont2023adaptive} on improving Shor-style measurements \cite{shor1996fault} for general codes based on consecutive syndrome information, \cite{anker2025compressing} on combining multiple stabilizers into one and re-encoding syndrome information using a classical code, \cite{berthusen2025adaptive} on removing unnecessary stabilizer measurements using code concatenation.

All schemes we have mentioned so far are designed for the standard Pauli error model. There is also a need to study more realistic noise models. In neutral atom and superconducting quantum architectures, a type of non-Pauli error, called \emph{leakage}, can dominate and take the encoded quantum state out of the computational subspace. Leakage is harmful because it produces syndrome patterns that differ from those under the standard Pauli error model. Fortunately, recent studies like \cite{wu2022erasure, kubica2023erasure, chow2024circuit} have proposed 
techniques that convert leakage errors into erasures and atom losses, which let us now restrict our attention to a mixed-error model that consists of both Pauli errors and erasures.\footnote{In the remainder of this work, we use ``mixed error model'' to refer to this, and ``mixed error'' to denote a combination of erasures and Pauli errors.}

Our first and basic observation is that we should distinguish among \emph{qubit losses, erasures, and located Pauli errors}, 
which were used interchangeably in early studies on loss tolerance \cite{haselgrove2006trade, rohde2007error, lu2008experimental}. 

The erasure correction process should be decomposed as follows, unless the architecture does not support qubit replacement: 
\begin{enumerate}
    \item Erasures are firstly converted to a located Pauli error via qubit replacements and further stabilizer measurements. \footnote{A successful stabilizer measurement on qubits, some of which are replaced by fresh ones, either extracts a syndrome bit that is at least valid at that point, or projects the state one step closer to the stabilizer subspace while simultaneously refreshing the syndrome bit corresponding to that stabilizer element.}
    \item When we obtain a usable syndrome string and apply the corresponding correction operator, both the located Pauli error and the unlocated Pauli error should be corrected, as long as they stay in the correctable region. \footnote{Whenever possible, erasures should be handled 
    during the computation to ensure fault tolerance and preserve the code distance. The \emph{$k$-shift erasure recovery} method recently proposed in \cite{kobayashi2025erasure} can be viewed as a nonadaptive strategy to periodically refresh the physical qubits and prevent erasure accumulation.}
\end{enumerate}

For the first step, where erasures are converted into located Pauli errors, a similar question then naturally arises:

\emph{What is the minimum number of stabilizer measurements required to convert an erasure into a located Pauli error?
}

This seems to stay unnoticed until \cite{matsumoto2025reducing} showed that for correcting erasures in quantum error correction, such a reduction exists. However, to distinguish and correct the unlocated Pauli error, a full set of generators must still be measured for quantum error correction. Here we rephrase 
this question in the fault-tolerant regime and discuss how the losses of different types of qubits affect the syndrome information and the state, and how to efficiently capture usable syndrome information. By abuse of notation, and to avoid confusion with the already occupied term \emph{erasure conversion}, we refer to the process of qubit replacement followed by projection as \emph{erasure error correction (EEC)}, although at this stage the erasures are not truly corrected, but merely converted into located Pauli errors. For qubits and prime-dimensional qudits, we link this minimal cost to the dimension of the stabilizer subgroup, while the measurements are found using the canonical form of stabilizer generators given a bipartite partition; for composite-dimensional qudits, we explicitly construct the canonical form of bipartite stabilizer states using results from \cite{sarkar2024qudit}.    

For the second step, we also show how to reconstruct a unique correction operator from the erasure pattern and a usable syndrome string, up to stabilizer operations.

Meanwhile, on the basis of recent advances in loss detection units and loss-tolerant syndrome extraction gadgets developed in \cite{chow2024circuit, perrin2025quantumerrorcorrectionresilient, baranes2025leveragingatomlosserrors}, we now revisit the entire syndrome extraction process and ask:

\emph{How to adaptively perform the syndrome measurements under the mixed error model?}

Adaptive FTEC schemes for Pauli errors cannot be directly applied to this model, because the syndrome patterns no longer stay the same. 
Existing loss-tolerant schemes rely on specific codes or architectures, or focus on designing loss-tolerant decoders \cite{rohde2007error, stace2009thresholds,barrett2010fault,fortescue2014fault,baker2021exploiting,bartolucci2023fusion, coble2025correction, gu2025fault,bartolucci2025comparison,kobayashi2025erasure}. 
However, to the best of our knowledge, for general codes, this question has not been formalized yet, and no prior work addresses it, even in the non-adaptive setting.

Adaptivity is again preferred to reduce the time overhead. It is also possible to adaptively select the next stabilizer measurement, as mentioned in \cite{berthusen2025adaptive}; this is in fact the only form of adaptivity required in our protocol. We expect that not only the syndrome information but also the erasure information can be leveraged here. Motivated by this, we integrate adaptive erasure error correction into the entire measurement sequence and design syndrome measurement protocols, generalizing the framework of \cite{tansuwannont2023adaptive} from the standard Pauli error model to the mixed error model.

This paper is organized as follows: Section~\ref{sec:adaptive_eec_related} reviews related work. In Section~\ref{sec:prelim}, we establish some preliminaries. In Section~\ref{sec:adaptive_eec}, we formulate adaptive erasure error correction: for prime-dimensional qudits, we reduce the problem of minimizing the number of additional stabilizer measurements to a subgroup-dimension problem and study certain syndrome extraction gadgets; for arbitrary composite dimensions, we provide the explicit construction of a canonical generating set under bipartition. In Section~\ref{sec:correction}, we show how to find a correction operator based on a usable syndrome string and the erasure pattern, and generalize the weak and strong FTEC conditions from the standard Pauli error model to the mixed error model. In Section~\ref{sec:adaptive_protocol}, we present adaptive syndrome-measurement protocols in the mixed model.

\subsection{Related work\label{sec:adaptive_eec_related}}

\paragraph{Adaptive erasure error correction.}
\cite{berthusen2025adaptive} uses a similar approach to reduce the number of measurements for the Pauli-error model. The protocol is designed for concatenated codes, of which the inner code is a small error-detection code that helps locate error locations in the outer code. Stabilizer generators supported only on blocks with no errors do not need to be measured.

Recently, \cite{matsumoto2025reducing} showed that the minimum number of measurements for correcting erasures is linked to the dimension of the subgroup of the stabilizer group that is supported only on remaining data qubits, motivated by quantum local recovery codes. However, how to correct erasures in a fault-tolerant way has not been addressed. In fault-tolerant syndrome extraction gadgets, the problem should be stated in a Hilbert space larger than the code space. We have also noticed that the key idea of reducing the number of stabilizer measurements is in fact more related to the well-known fact about the correctability of erasure errors, as discussed in early work like \cite{delfosse2012upper}.  In \cite{matsumoto2025reducing}, Remark 2 briefly points out the possibility of the distinction between qudits of prime dimension and those whose dimension is a prime power; however, qudits of composite dimension are not studied in that work. Remark 6 explains that by measuring the reduced subset of generators, only erasures can be corrected; the rest of the generator measurements are still required to distinguish Pauli errors even in quantum error correction. This, in fact, indicates that the advantage of erasure error correction does not hold under the mixed error model by only reducing the number of measurements, unless we carefully design a protocol that can handle both Pauli errors and erasures. Their results are therefore not directly comparable to this section, and we will further discuss how to design the entire FT syndrome measurement protocol.

\paragraph{Mixed error correctability.}
The textbook criterion for error correctability states that an $[[n,k,d]]$ stabilizer code can correct up to $t=\lfloor (d-1)/2 \rfloor$ Pauli errors and up to $d-1$ erasures, respectively. 
However, little attention has been drawn to the correctability of errors that are combinations of erasures and Pauli errors. An early theoretical result \cite[Theorem 2]{haselgrove2006trade} states that the correctability of $t$ Pauli errors of a stabilizer code is equivalent to the correctability of all possible combinations of $2m$ erasures and $t-m$ Pauli errors for any integer $0\leq m\leq t$. The FTEC conditions can be modified according to this result, and our syndrome extraction protocols will be designed based on it.

\paragraph{Shor-style sequences.} In the standard Pauli error model, if no internal Pauli fault occurs, measuring a full set of generators yields the right syndrome string $s$, as well as the right correction operator $\mathcal{E}(s)$. However, internal Pauli faults can cause the syndrome extraction to fail. The traditional Shor-style syndrome extraction \cite{shor1996fault} performs multiple rounds of syndrome string extraction, until the syndrome strings from $t+1$ consecutive rounds are identical, where $t=\lfloor (d-1)/2\rfloor $ is the maximum weight of the Pauli error that can be corrected by a distance-$d$ code. We say such a syndrome string is \emph{usable}, since the corresponding round of syndrome extraction is guaranteed to be unaffected by internal Pauli faults. In the worst case, a single fault occurs every $t+1$ rounds, and $(t+1)^2$ rounds of syndrome string extraction are required.

Adaptive FTEC protocols satisfying the strong and weak FTEC conditions for a stabilizer code of any distance are given in \cite{tansuwannont2023adaptive}, with the maximum number of rounds being $(t+3)^2/4-1$ and $(t+3)^2/4-2$, respectively. In these protocols, the difference between syndrome strings from two consecutive rounds is used to identify a round in which no Pauli fault occurs. Following this line of work, we extend and adapt the idea to the mixed error model.

\section{Preliminaries\label{sec:prelim}
}
\subsection{Stabilizer states and stabilizer codes}

An $n$-qubit system is defined on a $2^n$-dimensional complex Hilbert space $({\bC}^{2})^{\otimes n}$. A set of $n$-qubit commuting Pauli operators $\{g_i\}_{i=1}^r$ generates a group $\mathcal{S}$, called the \emph{stabilizer group}. Any $n$-qubit state $|\psi\rangle$ satisfies $g_i |\psi\rangle = |\psi\rangle$ for all $i$ is a \emph{stabilizer state}. $|\psi\rangle$ is \emph{stabilized} by the operators $\{g_i\}$. Equivalently speaking, $|\psi\rangle$ is on the joint $+1$ eigenspace of all generators. When $r=n$ and $\{g_i\}$ is a minimal generating set, the stabilizer group uniquely determines a stabilizer state.

Any such stabilizer group $\mathcal{S}$ that does not contain $-I$ defines an $n$-qubit stabilizer code $\mathcal{C}$. The normalizer of $\mathcal{S}$, denoted by $N(\mathcal{S})$, is the group formed by all $n$-qubit Pauli operators that commute with all elements in $\mathcal{S}$. $N(\mathcal{S})\setminus \mathcal{S}$ consists of all nontrivial logical operators of $\mathcal{S}$. The choices for nontrivial logical operators are not unique. Logical operators can be classified into cosets $[P]\coloneq P\mathcal{S} = \{PS: S\in \mathcal{S}\}$, and any two physical representatives $P_1, P_2 \in [P]$ that belong to the same coset are equivalent, considering their action on $\mathcal{C}$. Any two Pauli operators $E, F$ are equivalent under stabilizer operations, $E\sim F$, if there exists a stabilizer $S\in \mathcal{S}$ satisfying that $ES=F$ \cite{gottesman2009introduction}.

These rules can be extended to qudits (see \cite{gottesman1998fault, gheorghiu2014standard, sarkar2024qudit, eczoo_qudit_stabilizer} for reference). An $n$-qudit system is defined on a $d^n$-dimensional complex Hilbert space $({\bC}^d)^{\otimes n}$. 
Correspondingly, the generalized single-qudit Pauli group is generated by $X_d = \sum_{i=0}^{d-1} |i+1\rangle\langle i|$ and $Z_d = \sum_{i=0}^{d-1}\omega^i |i\rangle\langle i|$, where $\omega = \exp{\frac{2\pi i}{d}}$, the $d$-th root of unity:
\[
\overline{\mathcal{P}}_d = \{\omega^a X_d^rZ_d^s: 0\leq r,s\leq d-1\}.
\]

Elements of $\overline{\mathcal{P}}_d$ satisfy the commutation relation
\[
[ X_d^r Z_d^s, X_d^t Z_d^u ] = \omega^{st-ru}
I.
\]

where $[\cdot]$ is the commutator from group theory. Following the notations of \cite{sarkar2024qudit}, $\llbracket X_d^r Z_d^s, X_d^t Z_d^u   \rrbracket$ refers to $st-ru$, the phase exponent. $\llbracket X_d^r Z_d^s, X_d^t Z_d^u   \rrbracket = 0$ iff the two elements commute, which may be more familiar to readers who are used to the Pauli operator commutator $[A,B]:=AB-BA$. 

$n$-qudit Pauli products take the form
\[
\omega^a X^{\vec{r}}Z^{\vec{s}} = \omega^a X_1^{r_1} Z_1^{s_1} \otimes X_2^{r_2} Z_2^{s_2} \otimes \cdots \otimes X_n^{x_n} Z_n^{z_n}
\]

and form the generalized $n$-qudit Pauli group $\overline{\mathcal{P}}_{d,n}$. If we ignore the phase factor, the set of all $n$-qudit Pauli operators $\mathcal{P}_{d,n}$ has one-to-one correspondence to the quotient group $\overline{\mathcal{P}}_{d,n}/ \{\omega^i I:i\in \mathbb{N
} ,0\leq i\leq d-1 \}$.

The commutation relation for any two elements of $\overline{\mathcal{P}}_{d,n}$ is 
\[
[\omega^a X^{\vec{r} } Z^{\vec{s}}, \omega^b X^{\vec{t}} Z^{\vec{u}} ] = \omega^{\vec{s}\cdot \vec{t} -\vec{r}\vec{u}}  I
\]
where $I$ is always the identity matrix on the entire system.

For qudits with prime dimension, the rank of a Pauli subgroup is simply the size of its basis. In arbitrary composite dimensions, the structure is more naturally formulated in the language of modules, where the ``rank'' of the submodule is instead defined as the size of its minimal generating set. 
A stabilizer group $\mathcal{S}$ is a subgroup that does not contain $\{\omega^a I : a\neq 0\}$. We direct the interested reader to \cite{sarkar2024qudit} for a detailed discussion of qudits of composite dimension, as well as the ring theory background and other notations we shamelessly adopt from their work but do not restate here due to limited time.  

The Hadamard gate is generalized to the Fourier gate
\[
F = \frac{1}{\sqrt{d}} \sum_{i,j=0}^{d-1} \omega^{ij}|j\rangle \langle i|.
\]

The generalized CNOT gate, namely the SUM gate in \cite{gottesman1998fault}, \raisebox{-.1cm}{\includegraphics[scale=0.6]{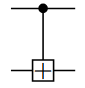}} is defined by
\[
\operatorname{CNOT} = \sum_{i,j=0}^{d-1} |i\rangle\langle i|\otimes |(i+j) \operatorname{mod} d\rangle \langle j|.
\]

We will omit the subscript $d$. 

If $\{g_i\}$ is a minimal generating set of $\mathcal{S}$, the number of logical qubits that $\mathcal{C}$ can encode is given by $k=n-r$. The distance $d$ of code $\mathcal{C}$ is further defined as the minimum weight of a nontrivial logical operator. This qubit stabilizer code is denoted as $[[n,k,d]]$. Correspondingly, a qudit stabilizer code with such parameters is denoted as $[[n,k,d]]_q$.

A minimal set of generators can be arranged as rows of a matrix, forming a \emph{stabilizer tableau}, a matrix representation of quantum states and operations. A code state of $\cC$ is stabilized by $n-k$ rows of generators. 

To handle phase factors, we may define the quotient group $\mathcal{S}/K $ with $K:=\{\omega^aI:0\leq a\leq q-1\}$ by ignoring the phases of the generators, and $K\mathcal{S}:=\{\omega^ag: g\in \mathcal{S}, 0\leq a\leq q-1\}$. For qubits and prime-dimensional qudits, $K\mathcal{S}$ is a group.

\begin{definition}
    The local subgroup of $\mathcal{S}$ on a given subset $L\subset [n]$ consists of all elements that are supported only on $L$ 
    \[
    \mathcal{S}^L := \{s\in \mathcal{S}: \operatorname{supp} s \subseteq  L\},
    \]
    where $\operatorname{supp}s$ is the support of a stabilizer $s\in \mathcal{S}$, namely the set of qub(d)its where $s$ acts nontrivially.

    $\mathcal{S}^\emptyset$ is trivial and $\dim(\mathcal{S}^\emptyset) = 0$.
\end{definition}

Two important tools used in this work are the canonical form of stabilizer generators and the cleaning lemma. 

\begin{definition}[The canonical form of generators for a stabilizer group with respect to a given bipartition \cite{fattal2004entanglement}]\label{def:canonical} 
For a bipartite 
state stabilized by $\mathcal{S}$ on the Hilbert space $\mathcal{H}_A\otimes \mathcal{H}_B$, $\mathcal{S}$ can be split into two local subgroups $\mathcal{S}^A$, $\mathcal{S}^B$, and a remaining subgroup $\mathcal{S}^{AB}$ whose elements are not contained in either local subgroup. Corresponding generator sets of these subgroups in a canonical form can be written as
\[
    \mathcal{S}^A = \langle a_i \otimes I_B\rangle , \mathcal{S}^B = \langle I_A \otimes b_j \rangle, 
    \mathcal{S}^{AB} = \langle g_k \otimes \overline{g}_k, h_k \otimes \overline{h}_k\rangle,
\]
where $\mathcal{S}^A$ (or $\mathcal{S}^B$) is the local subgroup 
on $A$ (or $B$), $\mathcal{S}^{AB}=\mathcal{S}\setminus (\mathcal{S}^A\sqcup \mathcal{S}^B)$ is generated by $\{(g_k \otimes \overline{g}_k, h_k \otimes \overline{h}_k)\}$, generators in pairs such that each of $g_k \otimes \overline{g}_k, h_k \otimes \overline{h}_k$ act nontrivially on both subsystems, and their local components anticommute on both $A$ and $B$, while commuting with all other generators of $\mathcal{S}$.

$\mathcal{S}$ is generated by 
\[
    \mathcal{S} = \langle  a_i \otimes I_B, I_A \otimes b_j, g_k \otimes \overline{g}_k, h_k \otimes \overline{h}_k\rangle. 
\]
\end{definition}

By abuse of notation, we also use $\mathcal{S}^L$ to denote its restriction on $L$ or its extension to a larger Hilbert space, depending on the context. We may avoid explicitly choosing a set of generators for the local groups by writing the stabilizer tableau as follows:

    \[
        \begin{array}{c}
    \overbracket{\phantom{sssss}}^{A} \phantom{ss}\overbracket{\phantom{sssss}}^{B}\\
    \begin{array}{ccc}
      \mathcal{S}^A & &I\\
      &\mathcal{S}^{AB}  & \\
      I &  & \mathcal{S}^B 
    \end{array}
    \end{array}
    \]

This was provided in the qubit case, but naturally extends to qudits of prime dimension. The authors of \cite{fattal2004entanglement} also remark that a generalization to qudits of composite dimension should be possible. There has been a line of related work\cite{fattal2004entanglement, zhong2021entanglement, looi2011tripartite} on bipartite or multipartite canonical decomposition for qubits, prime and square-free dimensional qudits. However, to our knowledge, an explicit construction of canonical generating sets for qudits of arbitrary composite dimension based on the partition has not appeared in the literature, and prior work has encountered difficulties. As recently highlighted in \cite{sarkar2024qudit}, when the composite dimension is not square-free, the maximal collections of non-commuting pairs fail to generate the full Pauli group. This helps explain why group-theoretic construction methods meet difficulties in the non-square-free composite case. We quote the following lemma that we will use to provide such a construction.

\begin{lemma}[Lemma~6.11 in \cite{sarkar2024qudit}]\label{lem:qudit6.11}

Let $S=\{s_0,s_1,\dots,s_{k-1}\}\subseteq\overline{\mathcal{P}}_n $ and let $A_{ij}=\llbracket s_i,s_j\rrbracket_d$ define a matrix $A\in\mathbb{Z}_d^{k\times k}$. The group $G:=\langle S\rangle$ has a Gram-Schmidt generating set $S_1\cup S_2\cup U$ where $|S_1|=|S_2|=\Theta(M_A)/2$, $M_A$ is the submodule of $\mathbb{Z}_d^k$ generated by the columns of $A$, and moreover there does not exist a Gram-Schmidt generating set with smaller $S_1$ or $S_2$.
    
\end{lemma}

The construction of the Gram-Schmidt generating set can be found in the proof of this lemma in \cite{sarkar2024qudit}.

\begin{lemma}[The cleaning lemma \cite{bravyi2009no} 
]\label{lemma:cleaning}
    Suppose the code $\cC$ has at least one nontrivial logical operator and $U\subset [n]$ is an arbitrary subset of qubits. One of the following is true: 
    \begin{enumerate}
        \item There is a nontrivial logical operator $P\in N(\mathcal{S})\setminus \mathcal{S}$ that is supported entirely in $U$, or
        \item for any $[P]\in N(\mathcal{S})/\mathcal{S}$, there exists a physical representative $P'\in [P]$ such that $P'$ acts trivially on qubits of $U$. $PP'$ can be written as a product of stabilizer generators $\prod_i s_i$ such that for each $i$, the support of $s_i$ overlaps with $U$. 
    \end{enumerate}
\end{lemma}
A direct consequence of this lemma is that for any subset $U$ that refers to a correctable set of erasures, any logical operator $P$ can be cleaned off of $U$, namely, there exists some $P' \in [P]$ supported only on $[n]\setminus U$. The cleaning lemma also extends to qudits.

\subsection{Quantum error correction and fault tolerance}

Suppose we are working with $q$-dimensional qudits, where $q\geq 2$ and $q$ is prime, thus including the qubit case. 
The joint $+1$ eigen subspace is called the \emph{ideal syndrome subspace} or the \emph{zero syndrome subspace}. Given a minimal set of generators of the stabilizer group $\mathcal{S}$, the labels of the eigen subspaces, which are numbers from the $q$-th root of unity $\{\omega^a: 0\leq a\leq q-1\}$, form a syndrome string. Most frequently people use $a$ to represent a syndrome qit, but we will be flexible about how the classical syndrome information is represented. A nontrivial Pauli error $E$ maps a code state to a non-ideal syndrome subspace.

The choice of the subspaces depends on the way to obtain $s$. 
Any minimal generating set of $\mathcal{S}$ uniquely defines a homomorphic syndrome map $\sigma: \mathcal{P}_n\rightarrow \mathbb{F}_q^{n-k}$ for an $[[n,k,d]]_q$ code $\mathcal{C}$. 

Suppose we use the minimal generating set $\{g_1, \cdots, g_{n-k}\}$, the syndrome map $\sigma(\cdot)$ is defined as 
\[
\sigma(E) = (s_1, s_2, ..., s_{n-k}),
\]
where 
$s_i\in \{\omega^a : 0\leq a\leq q-1\}$  
represents the corresponding eigenvalue of $g_i$'s eigen subspace.

For future use, we define $\vec{s}_1+ \vec{s}_2$ as 
the bitwise addition modulo $q$ for higher-dimensional systems, and $\vec{s}_1 \circ \vec{s}_2$ as the concatenation of $\vec{s}_1$ and $\vec{s}_2$.

Different Pauli operators may produce the same syndrome strings. We can decompose any Pauli error $E$ as a product of three terms
\[
E = L \cdot S\cdot D_{\vec{s}},
\]
where $L\in N(\mathcal{S})$ is a logical operator, $S\in \mathcal{S}$ is a stabilizer and $D_{\vec{s}}$, called a \emph{destabilizer} (or sometimes \emph{reduced error}) satisfies that $\sigma(D_{\vec{s}}) = \sigma(E) = \vec{s}$. $D_{\vec{s}}$ is the Pauli operator that has the minimum weight among all Pauli operators that produce the syndrome $\vec{s}$ and is uniquely determined by $\vec{s}$. Regarding the first two terms, $L$ can be a nontrivial logical error but always goes undetected; $S$ is a stabilizer and hence is never a Pauli error. Both produce trivial syndromes. Therefore, errors can be divided into equivalence classes according to $s$. The equivalence relation is denoted by $\sim$. We can define the reduced error map $\mathcal{E}(\cdot ): \mathbb{F}_q^{n-k}\rightarrow \mathcal{P}_n$ as
\[
\mathcal{E}(\vec{s}) = \operatorname*{arg\,min}_{E \in \mathcal{P}_n / \{\pm1,\pm i\}}
\left\{ \mathrm{wt}(E) : \sigma(E) = \vec{s} \right\} = D_s,
\]
where $\mathrm{wt}(E)$ denotes the Pauli weight of $E$. We will use $\mathcal{E}(\vec{s})$ and $D_{\vec{s}}$ interchangeably.

Here are some natural properties of the syndrome maps and the reduced error maps:
\begin{definition}
    Given a fixed set of generators and corresponding maps $\mathcal{E}, \sigma$, for any two syndrome strings $\vec{s}_1, \vec{s}_2\in \mathbb{F}_q^n$, we have 
    \[
        \mathcal{E} (\vec{s}_1 + \vec{s_2} ) \sim \mathcal{E} (\vec{s}_1) \cdot \mathcal{E}(\vec{s}_2),
    \]
    and for any two Pauli errors $E_1, E_2$, we have
    \[
    \sigma(E_1\cdot E_2) =\sigma(E_1)+ \sigma(E_2). 
    \]
\end{definition}

The QEC process consists of measuring the syndrome $s$ and applying the correction operator $\mathcal{E}(s)$. Stabilizer measurements with ancillary qubits extracts the syndrome information without disturbing the code state. 
However, the process of performing quantum error correction can be faulty, and internal Pauli faults that occur on one qudit may propagate to other qudits, resulting in a higher-weight error on data qudits. 
A fault-tolerant QEC process must satisfy the requirement that internal faults do not propagate into high-weight errors; this requirement is formalized as \emph{fault-tolerant error correction (FTEC) conditions}.

There are two variants of FTEC conditions, strong FTEC conditions \cite{aliferis2005quantum} and weak FTEC conditions \cite{delfosse2020short}.
They are needed in different settings, 
depending on whether errors are assumed to be handled across multiple levels of code concatenation.

\begin{definition}[Strong FTEC conditions {\cite[Restatement of Definition 1]{tansuwannont2023adaptive}}]
An error correction protocol is \emph{strongly $t$-fault} for an $[[n,k,d]]$ stabilizer code ($t\leq \lfloor\frac{d-1}{2}\rfloor$) if the following two conditions are satisfied:
\begin{enumerate}
    \item \textbf{Error correction correctness property (ECCP)}: If the input Pauli error has weight $r$ and $s$ internal Pauli faults occur during the protocol with $r+s\leq t$, then an ideal decoder outputs the same codeword whether it is applied to the input state or to the output state of the protocol.
    \item \textbf{Error-correction recovery property (ECRP):} Regardless of the weight of the input Pauli error, if $s$ internal Pauli faults occur during the protocol with $s \le t$, then the output state differs from a valid codeword by a Pauli error of weight at most $s$.

\end{enumerate}
\end{definition}

\begin{definition}[Weak conditions for fault-tolerant error correction {\cite[Restatement of Definition 2]{tansuwannont2023adaptive}}]
An error correction protocol is \emph{weakly $t$-fault tolerant} for an $[[n,k,d]]$ stabilizer code ($t\leq \lfloor\frac{d-1}{2}\rfloor$) if the following two conditions are satisfied:
	\begin{enumerate}
		\item \textbf{ECCP:} If the input Pauli error has weight $r$ and $s$ internal Pauli faults occur during the protocol with $r+s\leq t$, then an ideal decoder outputs the same codeword whether it is applied to the input state or to the output state of the protocol.
		\item \textbf{ECRP:} If the input Pauli error has weight $r$ and $s$ internal Pauli faults occur during the protocol with $r+s\leq t$, the output state differs from a valid codeword by a Pauli error of weight at most $s$.
	\end{enumerate}
	
	\label{def:weak_FT}
\end{definition}

However, these conditions do not directly apply to the mixed error model. We will later construct new conditions based on the following result about the mixed error correctability from \cite{haselgrove2006trade}.

\begin{definition}
    For any stabilizer code, 
    let $(e,p)$ be a pair of nonnegative integers representing a combination of $e$ erased qubits and $p$ single-qubit Pauli errors. We define the weight of such a combination as 
    \[\operatorname{wt}(e,p) = \frac{1}{2}e+p.\]
\end{definition}

Note that for a Pauli operator $E$, $\operatorname{wt}(E)$ still refers to its weight.

\begin{theorem}[The equivalence between the Pauli error correctability and the mixed error correctability {\cite[Restatement of Theorem 2]{haselgrove2006trade}}]\label{thm:trade-off-mixed-correctability}
    For any quantum error correcting code $\cC$, the following are equivalent:
    \begin{enumerate}
        \item $\cC$ can correct up to $t$ arbitrary single-qubit Pauli errors.
        \item $\cC$ can correct a combination of $e$ erased qubits and $p$ arbitrary single-qubit Pauli errors if $\operatorname{wt}(e,p)\leq t$, regardless of whether the Pauli errors act on the erased qubits.
    \end{enumerate}
\end{theorem}

\cite{haselgrove2006trade} was originally formulated for qubits; however, this theorem naturally extends to qudits, so we omit both the proof and its restatement in that setting.

\subsection{Losses, erasures and loss detection\label{sec:loss_erasures_and_loss_detection}}
Traditional QEC protocols have been mostly developed under the standard Pauli error model. In practice, realistic noise models include multiple non-Pauli error types, one of which is the loss of qubits or qudits. 
We use the terms \emph{detected loss} and \emph{erasure} interchangeably, since once the loss of a qudit is detected, its location is known.

There are two equivalent ways to mathematically model erasure errors and the replacement of erased qudits by fresh ancillas:
\begin{enumerate}
    \item A known subset of qudits is lost, replaced with an equal number of fresh ancillas. 
    \item Equivalently, the erased qudits are completely depolarized, and erasures may be modeled as a uniformly random Pauli error.
\end{enumerate}

However, as noted in Sec.~\ref{sec:Introduction}, we hold the opinion that it is better to treat losses, erasures, and located Pauli errors as three different types of errors. These three terms have been used interchangeably in early literature \cite{haselgrove2006trade, rohde2007error, lu2008experimental, ralph2005loss,stace2009thresholds,barrett2010fault, fortescue2014fault}.\footnote{Mostly because in the photonic community, ``loss tolerance'' is typically used to denote tolerance against heralded photon loss.} 
For terminological precision, here we identify \emph{detected losses} with erasures. Losses can be detected either intrinsically \cite{ralph2005loss} or via loss detection units \cite{perrin2025quantumerrorcorrectionresilient, baranes2025leveragingatomlosserrors}, depending on the architecture. 
Erasures and located Pauli errors should be distinguished as well, though they are considered equivalent in mathematical models \cite{grassl1997codes, wu2022erasure, sahay2023high}. We follow the first approach and represent the qudit replacement explicitly by adding extra rows and columns in the stabilizer tableau.

We also use several assumptions from \cite{perrin2025quantumerrorcorrectionresilient, baranes2025leveragingatomlosserrors}:
\begin{enumerate}
    \item Once a qudit is lost, all subsequent operations and measurements that act on this qudit are erased.
    \item The loss of a qudit is detected upon its measurement. A single-qudit measurement outcome is either an integer from $\{0,1,\cdots, d-1\}$, or a loss symbol.
    \item A qudit may be lost during entangling operations. 
\end{enumerate}

We refer the reader to these references for a more detailed analysis of loss detection units and loss-tolerant gadgets, including their fault tolerance to Pauli errors and the associated classical postprocessing.

The loss of unmeasured qudits has to be detected using methods such as applying loss detection units or modified circuits. There are two main types of loss detection units (LDUs): standard LDUs and teleportation-based LDUs. For syndrome extraction circuits, SWAP-based and teleportation-based modifications have been introduced in \cite{baranes2025leveragingatomlosserrors}.

\paragraph{Standard LDUs.} 
A standard loss detection unit (SLDU) is implemented using CNOT gates followed by an X gate on the data qubits. An SLDU requires two controlled gates, which increases the lifecycle of older qubits more than a teleportation-based LDU \cite{perrin2025quantumerrorcorrectionresilient}:

\begin{figure*}[htbp]
    \centering
\tikzset{every picture/.style={line width=0.75pt}} 
\resizebox{\textwidth}{!}{%
\begin{tikzpicture}[x=0.75pt,y=0.75pt,yscale=-1,xscale=1]

\draw    (345.25,73) -- (516,73) ;

\draw    (315.75,130) -- (462.5,130) ;

\draw    (389.75,65.75) -- (389.75,130) ;

\draw   (382.5,73) .. controls (382.5,69) and (385.75,65.75) .. (389.75,65.75) .. controls (393.75,65.75) and (397,69) .. (397,73) .. controls (397,77) and (393.75,80.25) .. (389.75,80.25) .. controls (385.75,80.25) and (382.5,77) .. (382.5,73) -- cycle ;

\draw  [fill={rgb, 255:red, 0; green, 0; blue, 0 }  ,fill opacity=1 ] (386.07,130) .. controls (386.07,127.97) and (387.72,126.32) .. (389.75,126.32) .. controls (391.78,126.32) and (393.43,127.97) .. (393.43,130) .. controls (393.43,132.03) and (391.78,133.68) .. (389.75,133.68) .. controls (387.72,133.68) and (386.07,132.03) .. (386.07,130) -- cycle ;

\draw  [fill={rgb, 255:red, 255; green, 255; blue, 255 }  ,fill opacity=1 ] (415.75,117) -- (444.75,117) -- (444.75,140) -- (415.75,140) -- cycle ;

\draw   (461,136.82) .. controls (461.09,127.79) and (467.85,120.5) .. (476.18,120.5) .. controls (484.44,120.5) and (491.16,127.66) .. (491.36,136.58) -- cycle ;

\draw    (463.5,146.5) -- (496.59,113.41) ;
\draw [shift={(498,112)}, rotate = 135] [color={rgb, 255:red, 0; green, 0; blue, 0 }  ][line width=0.75]    (10.93,-3.29) .. controls (6.95,-1.4) and (3.31,-0.3) .. (0,0) .. controls (3.31,0.3) and (6.95,1.4) .. (10.93,3.29)   ;

\draw    (477.68,74) -- (477.68,120.5)(474.68,74) -- (474.68,120.5) ;

\draw  [fill={rgb, 255:red, 255; green, 255; blue, 255 }  ,fill opacity=1 ] (461.75,60.5) -- (490.75,60.5) -- (490.75,83.5) -- (461.75,83.5) -- cycle ;

\draw    (635.25,74.5) -- (806,74.5) ;

\draw    (605.75,131.5) -- (752.5,131.5) ;

\draw    (708.75,139.18) -- (708.75,74.93) ;

\draw   (716,131.93) .. controls (716,135.93) and (712.75,139.18) .. (708.75,139.18) .. controls (704.75,139.18) and (701.5,135.93) .. (701.5,131.93) .. controls (701.5,127.92) and (704.75,124.68) .. (708.75,124.68) .. controls (712.75,124.68) and (716,127.92) .. (716,131.93) -- cycle ;

\draw  [fill={rgb, 255:red, 0; green, 0; blue, 0 }  ,fill opacity=1 ] (712.43,74.93) .. controls (712.43,76.96) and (710.78,78.61) .. (708.75,78.61) .. controls (706.72,78.61) and (705.07,76.96) .. (705.07,74.93) .. controls (705.07,72.9) and (706.72,71.25) .. (708.75,71.25) .. controls (710.78,71.25) and (712.43,72.9) .. (712.43,74.93) -- cycle ;

\draw  [fill={rgb, 255:red, 255; green, 255; blue, 255 }  ,fill opacity=1 ] (654.75,63.5) -- (683.75,63.5) -- (683.75,86.5) -- (654.75,86.5) -- cycle ;

\draw   (751,138.32) .. controls (751.09,129.29) and (757.85,122) .. (766.18,122) .. controls (774.44,122) and (781.16,129.16) .. (781.36,138.08) -- cycle ;

\draw    (753.5,148) -- (786.59,114.91) ;
\draw [shift={(788,113.5)}, rotate = 135] [color={rgb, 255:red, 0; green, 0; blue, 0 }  ][line width=0.75]    (10.93,-3.29) .. controls (6.95,-1.4) and (3.31,-0.3) .. (0,0) .. controls (3.31,0.3) and (6.95,1.4) .. (10.93,3.29)   ;

\draw    (767.68,75.5) -- (767.68,122)(764.68,75.5) -- (764.68,122) ;

\draw  [fill={rgb, 255:red, 255; green, 255; blue, 255 }  ,fill opacity=1 ] (751.75,62) -- (780.75,62) -- (780.75,85) -- (751.75,85) -- cycle ;

\draw    (59.25,72.5) -- (248.33,72.5) ;

\draw    (29.68,129.5) -- (245.33,129.5) ;

\draw    (103.75,65.25) -- (103.75,129.5) ;

\draw   (96.5,72.5) .. controls (96.5,68.5) and (99.75,65.25) .. (103.75,65.25) .. controls (107.75,65.25) and (111,68.5) .. (111,72.5) .. controls (111,76.5) and (107.75,79.75) .. (103.75,79.75) .. controls (99.75,79.75) and (96.5,76.5) .. (96.5,72.5) -- cycle ;

\draw  [fill={rgb, 255:red, 0; green, 0; blue, 0 }  ,fill opacity=1 ] (100.07,129.5) .. controls (100.07,127.47) and (101.72,125.82) .. (103.75,125.82) .. controls (105.78,125.82) and (107.43,127.47) .. (107.43,129.5) .. controls (107.43,131.53) and (105.78,133.18) .. (103.75,133.18) .. controls (101.72,133.18) and (100.07,131.53) .. (100.07,129.5) -- cycle ;

\draw  [fill={rgb, 255:red, 255; green, 255; blue, 255 }  ,fill opacity=1 ] (130.75,116.5) -- (159.75,116.5) -- (159.75,139.5) -- (130.75,139.5) -- cycle ;

\draw   (246,81.32) .. controls (246.09,72.29) and (252.85,65) .. (261.18,65) .. controls (269.44,65) and (276.16,72.16) .. (276.36,81.08) -- cycle ;

\draw    (248.5,91) -- (281.59,57.91) ;
\draw [shift={(283,56.5)}, rotate = 135] [color={rgb, 255:red, 0; green, 0; blue, 0 }  ][line width=0.75]    (10.93,-3.29) .. controls (6.95,-1.4) and (3.31,-0.3) .. (0,0) .. controls (3.31,0.3) and (6.95,1.4) .. (10.93,3.29)   ;

\draw  [fill={rgb, 255:red, 255; green, 255; blue, 255 }  ,fill opacity=1 ] (193.75,117) -- (222.75,117) -- (222.75,140) -- (193.75,140) -- cycle ;

\draw    (172.75,65.25) -- (172.75,129.5) ;

\draw   (165.5,72.5) .. controls (165.5,68.5) and (168.75,65.25) .. (172.75,65.25) .. controls (176.75,65.25) and (180,68.5) .. (180,72.5) .. controls (180,76.5) and (176.75,79.75) .. (172.75,79.75) .. controls (168.75,79.75) and (165.5,76.5) .. (165.5,72.5) -- cycle ;

\draw  [fill={rgb, 255:red, 0; green, 0; blue, 0 }  ,fill opacity=1 ] (169.07,129.5) .. controls (169.07,127.47) and (170.72,125.82) .. (172.75,125.82) .. controls (174.78,125.82) and (176.43,127.47) .. (176.43,129.5) .. controls (176.43,131.53) and (174.78,133.18) .. (172.75,133.18) .. controls (170.72,133.18) and (169.07,131.53) .. (169.07,129.5) -- cycle ;

\draw (422,120.9) node [anchor=north west][inner sep=0.75pt]    {$H$};

\draw (469.75,64.4) node [anchor=north west][inner sep=0.75pt]    {$Z$};

\draw (313.5,65.15) node [anchor=north west][inner sep=0.75pt]    {$| 0\rangle $};

\draw (284.5,119.9) node [anchor=north west][inner sep=0.75pt]    {$| \psi \rangle $};

\draw (520.5,65.15) node [anchor=north west][inner sep=0.75pt]    {$| \psi \rangle $};

\draw (477,25) node [anchor=north west][inner sep=0.75pt]   [align=left] {Teleportation-based LDU};

\draw (661,67.4) node [anchor=north west][inner sep=0.75pt]    {$H$};

\draw (603.5,66.65) node [anchor=north west][inner sep=0.75pt]    {$| 0\rangle $};

\draw (574.5,121.4) node [anchor=north west][inner sep=0.75pt]    {$| \psi \rangle $};

\draw (810.5,66.65) node [anchor=north west][inner sep=0.75pt]    {$| \psi \rangle $};

\draw (759.75,65.9) node [anchor=north west][inner sep=0.75pt]    {$X$};

\draw (91.75,25) node [anchor=north west][inner sep=0.75pt]   [align=left] {Standard LDU};

\draw (201.75,120.9) node [anchor=north west][inner sep=0.75pt]    {$X$};

\draw (137,120.4) node [anchor=north west][inner sep=0.75pt]    {$X$};

\draw (249.5,120.65) node [anchor=north west][inner sep=0.75pt]    {$| \psi \rangle $};

\draw (-1.5,119.4) node [anchor=north west][inner sep=0.75pt]    {$| \psi \rangle $};

\draw (27.5,64.65) node [anchor=north west][inner sep=0.75pt]    {$| 0\rangle $};

\end{tikzpicture}
}
    \caption{The circuit representation for the standard loss detection unit and the teleportation-based loss detection unit \cite{perrin2025quantumerrorcorrectionresilient}.}
    \label{fig:enter-label}
\end{figure*}
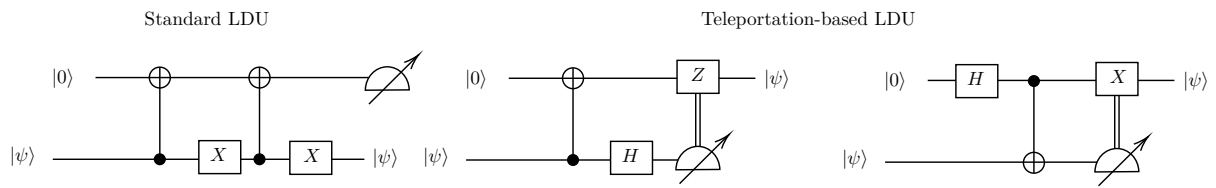

\paragraph{Teleportation-based  LDUs.} A teleportation-based LDU requires only one CNOT gate. 
If no loss occurs, the older qubit is swapped with a fresh ancilla. The measurement outcome is completely random and possible Pauli errors introduced by the measurement will be corrected later.  There are two equivalent circuit implementations for a TLDU, depending on whether the older qubit acts as the control or the target of the CNOT gate. For simplicity, we assume that the older qubit is always the target.

We can also trivially construct the teleportation-based LDU for qudits.

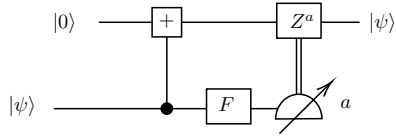
\begin{figure}[htbp]
    \centering
\tikzset{every picture/.style={line width=0.75pt}} 
\resizebox{0.33\textwidth}{!}{%
\begin{tikzpicture}[x=0.75pt,y=0.75pt,yscale=-1,xscale=1]

\draw    (366.25,120) -- (537,120) ;

\draw    (336.75,177) -- (483.5,177) ;

\draw    (410.75,112.75) -- (410.75,177) ;

\draw  [fill={rgb, 255:red, 0; green, 0; blue, 0 }  ,fill opacity=1 ] (407.07,177) .. controls (407.07,174.97) and (408.72,173.32) .. (410.75,173.32) .. controls (412.78,173.32) and (414.43,174.97) .. (414.43,177) .. controls (414.43,179.03) and (412.78,180.68) .. (410.75,180.68) .. controls (408.72,180.68) and (407.07,179.03) .. (407.07,177) -- cycle ;

\draw  [fill={rgb, 255:red, 255; green, 255; blue, 255 }  ,fill opacity=1 ] (436.75,164) -- (465.75,164) -- (465.75,187) -- (436.75,187) -- cycle ;

\draw   (482,183.82) .. controls (482.09,174.79) and (488.85,167.5) .. (497.18,167.5) .. controls (505.44,167.5) and (512.16,174.66) .. (512.36,183.58) -- cycle ;

\draw    (484.5,193.5) -- (517.59,160.41) ;
\draw [shift={(519,159)}, rotate = 135] [color={rgb, 255:red, 0; green, 0; blue, 0 }  ][line width=0.75]    (10.93,-3.29) .. controls (6.95,-1.4) and (3.31,-0.3) .. (0,0) .. controls (3.31,0.3) and (6.95,1.4) .. (10.93,3.29)   ;

\draw    (498.68,121) -- (498.68,167.5)(495.68,121) -- (495.68,167.5) ;

\draw  [fill={rgb, 255:red, 255; green, 255; blue, 255 }  ,fill opacity=1 ] (482.75,107.5) -- (511.75,107.5) -- (511.75,130.5) -- (482.75,130.5) -- cycle ;

\draw  [fill={rgb, 255:red, 255; green, 255; blue, 255 }  ,fill opacity=1 ] (401.1,110.35) -- (420.4,110.35) -- (420.4,129.65) -- (401.1,129.65) -- cycle ;

\draw (334.5,112.15) node [anchor=north west][inner sep=0.75pt]    {$| 0\rangle $};

\draw (305.5,166.9) node [anchor=north west][inner sep=0.75pt]    {$| \psi \rangle $};

\draw (541.5,112.15) node [anchor=north west][inner sep=0.75pt]    {$| \psi \rangle $};

\draw (404.1,113.4) node [anchor=north west][inner sep=0.75pt]    {$+$};

\draw (523,169.4) node [anchor=north west][inner sep=0.75pt]    {$a$};

\draw (489.75,113.4) node [anchor=north west][inner sep=0.75pt]    {$Z^{a}$};

\draw (443,167.9) node [anchor=north west][inner sep=0.75pt]    {$F$};

\end{tikzpicture}

}
    \caption{The circuit representation for the qudit teleportation-based loss detection unit. 
    }
\end{figure}

In this work, we mainly focus on loss detection by applying teleportation-based LDUs (TLDUs) to data qudits undergoing controlled operations in a stabilizer measurement circuit. Its advantages over standard LDUs have been demonstrated in \cite{perrin2025quantumerrorcorrectionresilient}.

If at least one qudit is lost during a TLDU, the teleportation procedure fails, and one of the following three cases occurs, as shown in Fig. \ref{fig:3tldu_loss}:
\begin{enumerate}
    \item The older (data) qudit is lost. In this case, the loss is detected, and the older qudit is replaced by a fresh ancilla.
    \item The fresh qudit used for teleportation is lost and the loss goes undetected. The older qudit is measured and replaced by the lost fresh qudit. The loss will be detected in the next round of loss detection. 
    \item Both the older qudit and the fresh qudit are lost. The loss of the older qudit is detected, and it is replaced by another lost qudit. Further loss detection and erasure error correction are required. 
\end{enumerate}

It is straightforward to show that these cases yield the same ``unaffected'' local subgroup, which remains untouched and therefore need not be reconstructed in adaptive EEC. The ``affected'' part, whether punctured by incorrectly measuring the older data qudits or out of control due to qudit losses, is the only part that needs to be restored.

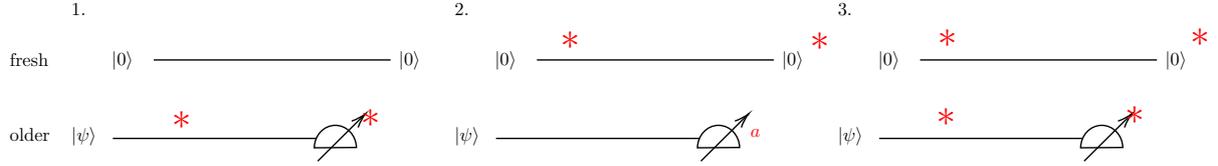
\begin{figure*}[htbp]
    \centering
\tikzset{every picture/.style={line width=0.75pt}} 
\resizebox{\textwidth}{!}{%

\begin{tikzpicture}[x=0.75pt,y=0.75pt,yscale=-1,xscale=1]

\draw    (107.25,59) -- (278,59) ;

\draw    (77.75,116) -- (224.5,116) ;

\draw   (223,122.82) .. controls (223.09,113.79) and (229.85,106.5) .. (238.18,106.5) .. controls (246.44,106.5) and (253.16,113.66) .. (253.36,122.58) -- cycle ;

\draw    (225.5,132.5) -- (258.59,99.41) ;
\draw [shift={(260,98)}, rotate = 135] [color={rgb, 255:red, 0; green, 0; blue, 0 }  ][line width=0.75]    (10.93,-3.29) .. controls (6.95,-1.4) and (3.31,-0.3) .. (0,0) .. controls (3.31,0.3) and (6.95,1.4) .. (10.93,3.29)   ;

\draw    (383.5,59) -- (554.25,59) ;

\draw    (354,116) -- (500.75,116) ;

\draw   (499.25,122.82) .. controls (499.34,113.79) and (506.1,106.5) .. (514.43,106.5) .. controls (522.69,106.5) and (529.41,113.66) .. (529.61,122.58) -- cycle ;

\draw    (501.75,132.5) -- (534.84,99.41) ;
\draw [shift={(536.25,98)}, rotate = 135] [color={rgb, 255:red, 0; green, 0; blue, 0 }  ][line width=0.75]    (10.93,-3.29) .. controls (6.95,-1.4) and (3.31,-0.3) .. (0,0) .. controls (3.31,0.3) and (6.95,1.4) .. (10.93,3.29)   ;

\draw    (659.75,59) -- (830.5,59) ;

\draw    (630.25,116) -- (777,116) ;

\draw   (775.5,122.82) .. controls (775.59,113.79) and (782.35,106.5) .. (790.68,106.5) .. controls (798.94,106.5) and (805.66,113.66) .. (805.86,122.58) -- cycle ;

\draw    (778,132.5) -- (811.09,99.41) ;
\draw [shift={(812.5,98)}, rotate = 135] [color={rgb, 255:red, 0; green, 0; blue, 0 }  ][line width=0.75]    (10.93,-3.29) .. controls (6.95,-1.4) and (3.31,-0.3) .. (0,0) .. controls (3.31,0.3) and (6.95,1.4) .. (10.93,3.29)   ;

\draw (46.5,16.5) node [anchor=north west][inner sep=0.75pt]   [align=left] {1.};

\draw (322.75,16.5) node [anchor=north west][inner sep=0.75pt]   [align=left] {2.};

\draw (599,16.5) node [anchor=north west][inner sep=0.75pt]   [align=left] {3.};

\draw (671,35.4) node [anchor=north west][inner sep=0.75pt]  [font=\huge,color={rgb, 255:red, 255; green, 0; blue, 0 }  ,opacity=1 ]  {$*$};

\draw (670,92.4) node [anchor=north west][inner sep=0.75pt]  [font=\huge,color={rgb, 255:red, 255; green, 0; blue, 0 }  ,opacity=1 ]  {$*$};

\draw (628,51.15) node [anchor=north west][inner sep=0.75pt]    {$| 0\rangle $};

\draw (599,105.9) node [anchor=north west][inner sep=0.75pt]    {$| \psi \rangle $};

\draw (835,51.15) node [anchor=north west][inner sep=0.75pt]    {$| 0\rangle $};

\draw (806,91.4) node [anchor=north west][inner sep=0.75pt]  [font=\huge,color={rgb, 255:red, 255; green, 0; blue, 0 }  ,opacity=1 ]  {$*$};

\draw (853,34.4) node [anchor=north west][inner sep=0.75pt]  [font=\huge,color={rgb, 255:red, 255; green, 0; blue, 0 }  ,opacity=1 ]  {$*$};

\draw (399,36.4) node [anchor=north west][inner sep=0.75pt]  [font=\huge,color={rgb, 255:red, 255; green, 0; blue, 0 }  ,opacity=1 ]  {$*$};

\draw (351.75,51.15) node [anchor=north west][inner sep=0.75pt]    {$| 0\rangle $};

\draw (322.75,105.9) node [anchor=north west][inner sep=0.75pt]    {$| \psi \rangle $};

\draw (558.75,51.15) node [anchor=north west][inner sep=0.75pt]    {$| 0\rangle $};

\draw (536,107.4) node [anchor=north west][inner sep=0.75pt]  [font=\normalsize,color={rgb, 255:red, 255; green, 0; blue, 0 }  ,opacity=1 ]  {$a$};

\draw (579,37.4) node [anchor=north west][inner sep=0.75pt]  [font=\huge,color={rgb, 255:red, 255; green, 0; blue, 0 }  ,opacity=1 ]  {$*$};

\draw (119,94.4) node [anchor=north west][inner sep=0.75pt]  [font=\huge,color={rgb, 255:red, 255; green, 0; blue, 0 }  ,opacity=1 ]  {$*$};

\draw (255,93.4) node [anchor=north west][inner sep=0.75pt]  [font=\huge,color={rgb, 255:red, 255; green, 0; blue, 0 }  ,opacity=1 ]  {$*$};

\draw (75.5,51.15) node [anchor=north west][inner sep=0.75pt]    {$| 0\rangle $};

\draw (46.5,105.9) node [anchor=north west][inner sep=0.75pt]    {$| \psi \rangle $};

\draw (282.5,51.15) node [anchor=north west][inner sep=0.75pt]    {$| 0 \rangle $};

\draw (2,107) node [anchor=north west][inner sep=0.75pt]   [align=left] {older};

\draw (2,53) node [anchor=north west][inner sep=0.75pt]   [align=left] {fresh};

\end{tikzpicture}

}
    \caption{Three possible loss scenarios in a teleportation-based LDU (TLDU). The red asterisks indicate qudit losses. All subsequent operations acting on a lost qudit, including single- and two-qudit gates, have no effect until the qudit is replaced by a fresh ancilla. A measurement performed on a lost qudit returns a loss symbol rather than a classical bit, and has no effect on the lost qudit. In Case $2$, by incorrectly measuring the older data qudit, a random and useless outcome $a$ is extracted. }
    \label{fig:3tldu_loss}
\end{figure*}  
 We expect similar results to hold for modified circuits that naturally detect loss without additional overhead, and we provide brief comments in the Discussion section.

\section{Adaptive erasure error ``correction''\label{sec:adaptive_eec}}
Detected losses become erasures. A common first step in erasure error correction is to replace erased qudits with fresh ancillas, or equivalently, to reinitialize them to a state unentangled with the rest of the system. A sequence of generator measurements is performed afterwards, simultaneously projecting the state back to the stabilizer subspace on the new set of registers and rewriting corresponding syndrome information. The erasure is then converted into a located Pauli error, which has not been corrected at this stage. However, we can at least claim that we have ``corrected'' the qudit losses and can continue working within the same stabilizer subspace. We also do not want to cause any confusion by reusing the term ``erasure conversion''. Here comes the name of this stage: erasure error correction. 

At this moment, we may restrict our attention to projecting the state into a stabilizer subspace. For a fixed generating set, the resulting state may lie in any syndrome subspace, and the random phase factors will be handled later. Instead of measuring the full set of stabilizer generators, one might expect that a much smaller number of stabilizer measurements is required when the number of erased qudits is small, or if the code is good enough \cite{eczoo_good_qldpc}. 
For qubits and prime-dimensional qudits, we show how to reduce the problem of finding the minimal set of stabilizers for adaptive EEC to identifying a local subgroup ``unaffected'' by erasures, and analyze specific syndrome extraction gadgets. In addition, we construct a canonical generating set for stabilizer groups of arbitrary composite-dimensional qudits under bipartition.

\subsection{Reduction to a subgroup-dimension problem for prime qudits\label{sec:adaptive_eec_subgroup}}
We begin by presenting some generic results for qubits and qudits of prime dimension.
\begin{theorem}\label{thm:min_msmt_general}
The minimum number of measurements required to convert the tensor product of a state stabilized by $\mathcal{S}$ on Hilbert space $\cH_A\otimes\cH_E$ and an ancillary state on $\cH_B$, into the tensor product of a state that has the same stabilizer generators as $\mathcal{S}'$ up to phases 
on $\cH_A\otimes \cH_B$
and a state on $\cH_E$ is 
 
        \[
    \dim \mathcal{S}' - \dim \big((\mathcal{S}^A/K) \cap (\mathcal{S}'^{A}/K)\big)-\dim \mathcal{S}'^B.
    \] 
\end{theorem}
\begin{proof} See Appendix~\ref{proof:min_msmt_general}. 
\end{proof}

Theorem~\ref{thm:min_msmt_general} applies to stabilizer states that are uniquely determined by the stabilizer group. We may allow some freedom and extend it to stabilizer codes. Cleaning lemma implies that to preserve the logical information, an additional assumption is needed: for any logical operator, there must exist a physical representative that does not require projection. We now address the problem in the following scenario:

Consider a state that was initially stabilized by $\mathcal{S}_{\cC}$. After undergoing quantum operations and erasures, the state now lies in the Hilbert space $\cH_A\otimes \cH_E$, where $A$ denotes the set of unerased data qudits and $E$ denotes the set of erased data qudits and ancillas. We say the erasure set $E$ is \emph{correctable}\footnote{Note that under the mixed error model, the term \emph{correctable} is more appropriately interpreted as \emph{not fully uncorrectable} in the presence of unknown Pauli errors. However, for consistency with existing terminology, we retain this term.} if the logical information can be cleaned off of $E$; that is, the logical operators of the code admit physical representatives supported only on $A$ \cite{bravyi2009no}. For an $[[n,k,d]]_q$ stabilizer code, the condition $\abs{A}\geq n-d+1$ is sufficient for correctability. In the following corollary, we restrict to this sufficient condition.

Let $B$ denote the set of fresh qudits that replace the erased data qudits. To reconstruct the code state on $\cH_A\otimes \cH_B$, the minimum number of measurements required to convert the erasures into a located Pauli error  is
\begin{corollary}\label{cor:min_msmt_code}
    Given an $[[n,k,d]]$ stabilizer code $\cC$, consider a state on 
    $\cH_A\otimes \cH_E$ stabilized by $\mathcal{S}$ and encoding unknown logical information. Assume $\abs{A}\geq n-d+1$. 

    Then the minimum number of measurements required to convert it to a state encoding the same logical information using $\cC$ on $\cH_A\otimes \cH_B$, up to some phases on the generators of $\mathcal{S}_{\mathcal{C}}$,
     is
        \[
    \dim \mathcal{S}_\cC - \dim \big((\mathcal{S}_\cC^A /K) \cap ( \mathcal{S}^A/K)\big) - \dim \mathcal{S}_\cC^{B}. 
    \]
    For nondegenerate codes, the last term $
    \dim \mathcal{S}_\cC^{B} =0 $.
\end{corollary}
\begin{proof}
See Appendix~\ref{proof:min_msmt_code}.
\end{proof}

In specific erasure error correction scenarios, 
the last term may be removed if there is no entanglement among the fresh ancillas, for example, when lost data qubits are replaced by ancillas supplied by TLDUs, or when maintaining a larger pre-prepared ancilla factory is impractical.

We defer the study of minimizing the total weight of the measured generators, as well as the characterization of larger correctable regions beyond the $d-1$ sufficient condition\cite{baspin2022connectivity, dai2025locality} to future work, since they are independent of the main focus here.

\subsection{Canonical form of bipartite stabilizer states for qudits of composite dimension}

For qudits of composite dimension, things are much more complicated. 
For the sake of completeness, we first provide a constructive proof for the canonical form of bipartite stabilizer states on qudits of arbitrary composite dimension.

As suggested in \cite{sarkar2024qudit}, one needs to work with the size of the minimal generating set, since the symplectic representation is treated as modules and the rank (or dimension) of its submodules is not well defined. 
 Following Lemma~\ref{lem:qudit6.11}, we construct the local and nonlocal generating sets in a manner similar to Definition~\ref{def:canonical}, that is, the nonlocal generating sets of a composite qudit bipartition do consist of pairs of generators whose local projections non-commute, but commute with all other generators. While this may appear straightforward, we have not yet noticed such a construction in the existing literature for arbitrary composite-dimensional qudits. 

\begin{definition}
    Let $p_A$ be the projection of the stabilizer group $\mathcal{S}\subseteq \overline{\mathcal{P}}_n$ onto a subset $A\subseteq[n]$, with the image denoted as $\mathcal{S}|_A = p_A(\mathcal{S})$. 
    We call $\mathcal{S}|_A$ the local restriction of $\mathcal{S}$ to $A$.
\end{definition}

By definition, $\ker (p_A) = \mathcal{S}^B $. Therefore, by the first isomorphism theorem, there exists an isomorphism
    \[
    \varphi_A: \mathcal{S}|_A\xrightarrow{\sim} \mathcal{S}/\mathcal{S}^B.
    \]

We borrow some notions from Lemma~\ref{lem:qudit6.11} and provide the construction in the following lemma.
\begin{lemma}
    Let $S=\{s_0, s_1, \dots, s_{k-1}\}$ be a generating set of a stabilizer group $\mathcal{S} = \langle S\rangle \subseteq \overline{\mathcal{P}}_n$. 
    Given a partition of qudits $\{A,B\}$, let $S|_A = \{s_0|_A, s_1|_A, \dots, s_{k-1}|_A \}$ and $S|_B = \{s_0|_B, s_1|_B, \dots, s_{k-1}|_B \}$ be the local restrictions of $S$ that generate $\mathcal{S}|_A$ and $\mathcal{S}|_B$ respectively. Let $C^A_{ij}=\llbracket s_i|_A,s_j|_A\rrbracket_d, C^B_{ij}=\llbracket s_i|_B,s_j|_B\rrbracket_d$ define matrices $C^A, C^B\in\mathbb{Z}_d^{k\times k}$.
    
    If one of the following statements is true:
    \begin{enumerate}
        \item $\mathcal{S}$ defines a stabilizer state, or
        \item $\mathcal{S}$ defines a stabilizer code, and $A$ is correctable,
    \end{enumerate}
    then there exists generating sets of the local subgroups $\mathcal{S}^A, \mathcal{S}^B$ and the nonlocal subgroup $\mathcal{S}^{AB}$ in the canonical form
\[
    S^A = \{ a_i \otimes I_B\} , S^B = \{ I_A \otimes b_j \}, 
    S^{AB} = \{ g_k \otimes \overline{g}_k, h_k \otimes \overline{h}_k\}_{k=1}^{r},
\]
  where the non-commuting pairs $(g_k \otimes \overline{g}_k, h_k \otimes \overline{h}_k)$ are obtained from the Gram-Schmidt generating set of $\mathcal{S}|_A$ given by Lemma~\ref{lem:qudit6.11}, and $r = \Theta(M_{C^A})/2$.
\end{lemma}
\begin{proof}

    Let $S_{A1}\sqcup S_{A2}\sqcup U_A$ be the Gram-Schmidt generating set of $ \mathcal{S}|_A$ constructed from $S$ by Lemma~\ref{lem:qudit6.11}. $S_{A1}, S_{A2}$ is a collection of non-commuting pairs that has minimum possible size among all Gram-Schmidt generating sets, and
    \[
    \abs{S_{A1}} = \abs{S_{A2}} = \frac{\Theta(M_{C^A})}{2}, 
    U_A\subseteq \mathcal{Z}(\mathcal{S}
    |_A).
    \]
    For simplicity, let
    \[
    r=\frac{\Theta(M_{C^A})}{2}, t=\abs{U_A}.
    \]

    The construction takes three steps:

    \paragraph{Step 1. Construct the nonlocal pairs from $\mathcal{S}|_A$.}

    For each non-commuting pair $(s_{A1}, s_{A2})$ from $S_{A1}, S_{A2}$, $\llbracket s_{A1}, s_{A2}\rrbracket \neq 0$. Choose lifts $s_1, s_2\in\mathcal{S}$ satisfying \[
    p_A(s_1) = s_{A1}, p_A(s_2) = s_{A2}.
    \]
    Write
    \[
    s_1 = s_{A1}\otimes s_{B1}, s_2 = s_{A2}\otimes s_{B2}
    \]
    for some $s_{B1}, s_{B2}$.
    $s_1, s_2\in \mathcal{S} $ requires that 
    \[
    \llbracket s_1, s_2 \rrbracket = \llbracket s_{A1}, s_{A2}\rrbracket + \llbracket s_{B1}, s_{B2} \rrbracket \equiv 0 \pmod d.
    \]
    Therefore $\llbracket s_{B1}, s_{B2} \rrbracket \neq 0$. 
    It's obvious that $s_{A1}, s_{A2}, s_{B1}, s_{B2}$ are nontrivial.
    
    Now suppose $s_1'=s_{A1}'\otimes s_{B1}', s_2' = s_{A2}'\otimes s_{B2}'$ are lifts of a distinct non-commuting pair $(s_{A1}', s_{A2}')$. Since $s_1, s_2, s_1', s_2'\in \mathcal{S}$, they commute with each other. By the definition of non-commuting pairs,
    \[
    \llbracket s_{A1}, s_{A1}' \rrbracket = \llbracket s_{A1}, s_{A2}' \rrbracket = \llbracket s_{A2}, s_{A1}'  \rrbracket = \llbracket s_{A2}, s_{A2}'  \rrbracket = 0.
    \]
    Accordingly we have 
    \[
    \llbracket s_{B1}, s_{B1}' \rrbracket = \llbracket s_{B1}, s_{B2}' \rrbracket = \llbracket s_{B2}, s_{B1}'  \rrbracket = \llbracket s_{B2}, s_{B2}'  \rrbracket = 0.
    \]

    Therefore, the lifts of all non-commuting pairs $(s_{A1}, s_{A2})\in S_{A1}\times S_{A2}$, $\{(s_1, s_2)\}$, form a collection of non-commuting pairs which have nontrivial components on both $A$ and $B$. W.l.o.g., we denote this collection of non-commuting pairs by $G^{AB} = \{g_k\otimes \overline{g}_k, h_k\otimes \overline{h}_k\}_{k=1}^r$.

    \paragraph{Step 2. Construct $S^A$ from $\mathcal{Z}(\mathcal{S}|_A)$.} We first define
    \[
    L_A:=\{g_A\in \mathcal{S}|_A : g_A\otimes I_B \in \mathcal{S}\}
    \]
    and claim that $L_A = \mathcal{Z} (\mathcal{S}|_A)$.

    It's straightforward that $L_A\subseteq \mathcal{Z}(\mathcal{S}|_A)$. Conversely, for any $u_A\in \mathcal{Z}(\mathcal{S}|_A)$, $s=s_A\otimes s_B\in \mathcal{S}$, since \[\llbracket u_A, s_A\rrbracket =0,\]
    we have 
    \[\llbracket u_A\otimes I_B, s_A\otimes s_B\rrbracket = 0.\]
    Hence $u_A\otimes I_B\in N(\mathcal{S})$. When $\mathcal{S}$ uniquely determines a stabilizer state, $N(\mathcal{S})=\mathcal{S}$; when $\mathcal{S}$ defines a stabilizer code and $A$ is correctable, $u_A\otimes I_B$ cannot be a physical representative of any nontrivial logical operator by the cleaning lemma. So $u_A\otimes I_B\in \mathcal{S}$ in both cases and $\mathcal{Z}(\mathcal{S}|_A) \subseteq L_A$.

    We now prove that $\mathcal{S}^A\cong\mathcal{Z}(\mathcal{S}|_A)$. Consider the projection $p: \mathcal{S}^A \rightarrow \mathcal{Z}(\mathcal{S}|_A)=L_A$ defined as $p(g_A\otimes I_B) = g_A$ for any $g_A\otimes I_B\in \mathcal{S}^A$. If $p(g_A\otimes I_B) = I_A$, $g_A=I_A$. Therefore, the kernel is trivial, and $p$ is injective. For any $g_A\in L_A$, by definition $g_A\otimes I_B\in \mathcal{S}$, and $g_A\otimes I_B\in \mathcal{S}^A$, so $p$ is surjective. Hence $p$ is bijective and $\mathcal{S}^A\cong\mathcal{Z}(\mathcal{S}|_A)$. Therefore any minimal generating set of $L_A$ lifts bijectively to a minimal generating set of $\mathcal{S}^A$. W.l.o.g., suppose we choose $\{u_1, \dots, u_m\}$ as the minimal generating set of $L_A$. $S^A$ is then
    \[\{u_i\otimes I_B\}_{i=1}^m.\]
    
    \paragraph{Step 3. Construct $S^B$ from the lifted Gram-Schmidt generating set.} Recall that the Gram-Schmidt generating set in Lemma~\ref{lem:qudit6.11} is constructed using an invertible matrix $L\in \mathbb{Z}_d^{k\times k}$. For each $s'_{i}\in S'_A = S_{A1}\sqcup S_{A2}\sqcup U_A$, 
    \[
    s_i' = \prod_j (s_j|_A)^{L^{-1}_{ij} } 
    \]
    where $s_j\in S$. $s'_{i}$ has a lift 
    \[
    \tilde{s}_i = \prod_j (s_j)^{L_{ij}^{-1} } 
    \]
    satisfying that
    \[
    p_A(\tilde{s}_i) = s_i'.
    \]
    Define a map $\Psi:\mathbb{Z}_d^{2r+t}\rightarrow \mathcal{S}|_A$ that generates $\mathcal{S}|_A$
    \[
    \Psi(\vec{x}) = \prod_{i=1}^{2r+t} {s_i'}^{x_i}.
    \]
    The generating set of the kernel
    \[
    K:=\ker (\Psi)
    \]
    can be computed explicitly. W.l.o.g., suppose its generating set is 
    \[
    \{\kappa_1, \dots, \kappa_l\}
    \]
    where each $\kappa_i=(\kappa_{i1}, \dots, \kappa_{i,2r+t})\in \mathbb{Z}_d^{2r+t}$.
    Consider the lift
    \[
    b_i = \prod_{j=1}^{2r+t} \tilde{s}_j^{\kappa_{ij}}
    \]
    Since $\kappa_i\in \ker (\Psi)$, 
    \[
    p_A(b_i) = I_A.
    \]
    Therefore, 
    \[
    b_i\in\ker (p_A)=\mathcal{S}^B.
    \]
    For any $b\in\mathcal{S}^B$, 
    \[
    b = \prod_{j=1}^{2r+t} \tilde{s}_j^{x_j}
    \]
    for some $\vec{x}\in \mathbb{Z}_d^{2r+t}$.
    Since $b\in \mathcal{S}^B=\ker(p_A)$, 
    \[
    I_A = p_A(b) = \prod_{j=1}^{2r+t} {s_j'}^{x_j}.
    \]
    So $\vec{x}\in K.$
    Therefore, $\{b_i\}_{i=1}^{l}$ generates $\mathcal{S}^B$ and a minimal generating set $S^B$ can be computed from this set.

    To complete the proof, we now show that $G^{AB}$ forms a minimal generating set of $\mathcal{S}^{AB}$. 
    Since $S'_A=S_{A1}\sqcup S_{A2}\sqcup U_A=\{s_i'\}$ generates $\mathcal S|_A$ and the corresponding lifts $\{\tilde{s}_i\}$ generate $\mathcal S$, for every $s_i'\in U_A$, the lift $\tilde s_i$ lies in $\langle S^A,S^B\rangle$ by the constructions in Steps 2 and 3, while for every $s_i'\in S_{A1}\sqcup S_{A2}$, its lift $\tilde{s}_i$ lies in $G^{AB}$. So $\mathcal S=\langle S^A,S^B,G^{AB}\rangle$ and $G^{AB}$ generates the nonlocal subgroup $\mathcal{S}^{AB}$. Since $S_{A1}\sqcup S_{A2}$ forms a minimal set of non-commuting pairs of $\mathcal S|_A$, $r$ is the minimum possible number of such pairs. Since $G^{AB}$ contains exactly $r$ pairs, it is a minimal generating set of $\mathcal{S}^{AB}$.
\end{proof}

By replacing the subgroup dimensions with the size of minimal generating sets of the local subgroup, Theorem~\ref{thm:min_msmt_general} and Corollary~\ref{cor:min_msmt_code} can be generalized to qudits of arbitrary dimensions. We defer the details to future work, as we narrow our focus to prime-dimensional cases starting in the next subsection.

\subsection{Reducing the number of measurements in specific syndrome extraction gadgets\label{sec:adaptive_eec_schemes}}
In this section, we present a detailed classification and analysis of loss scenarios in specific syndrome extraction schemes. We will focus on Shor-style SE circuits and defer the results for other methods to the next version. There are also different methods for handling qubit losses, depending on the architecture. We discuss the one that equips each circuit with data TLDUs. 
Our results easily extend to other methods via tedious calculations on the stabilizer tableaux.

\subsubsection{Notations}

A Shor-style SE circuit extracts one bit(qit) syndrome information at one time. A $w$-qudit ancilla cat state (or a $w$-partite qudit GHZ state) is used to extract the syndrome corresponding to a generator of weight $w$. The controlled gates are applied between data qudits and syndrome qudits transversally.

Suppose we are working with $q$-dimensional qudits, use the $[[n,k,d]]_q$ stabilizer code $\mathcal{C}$ and measure a stabilizer
\[
P = P_{a_1}\cdots P_{a_w},
\]
where $w$ is the weight of $P$, and each $P_{a_j}$ is the nontrivial single-qudit component acting on the data qudit $q_{a_j}$. For consistency, we label the corresponding syndrome qudit by $a_j$.

During the measurement, assume that a subset $L_1$ of data qudits and a subset $L_2$ of syndrome qudits are lost. Let
\[
L = L_1 \cup L_2
\]
denote the set of lost qudits, so that the corresponding components of $P$ are not measured. The stabilizer measurement circuit instead measures $P^L\sim P^{\overline{L}}$ on the code state. $L$ also represents the set of data qudits affected by erased controlled operations.

\begin{definition}[Affected data qudits]
A data qudit is \emph{affected} by qudit losses in a Shor-style SE circuit if at least one of this data qudit and its corresponding syndrome qudit is lost during stabilizer measurement.
\end{definition}
We write
\[
\overline{L_1} = [n] \setminus L_1
\quad\text{and}\quad
\overline{L} = [n]\setminus L. 
\]

The local stabilizer subgroup on $\overline{L_1}$ is mostly not affected by erasures. The remaining affected subgroup has to be reconstructed on the unerased data qudits and fresh ancillas; this is implemented by subsequent stabilizer measurements, which we may incorporate into the measurement sequence of the current syndrome string extraction round. 

We now formalize the minimum number of stabilizer measurements required for EEC in a Shor-style SE circuit as the following lemma:

\begin{lemma}[Minimal EEC cost in Shor-style SE]\label{lem:shor_eec}
    In a Shor-style SE circuit, let $n,k,d$ be the code parameters, $P, L_1, L_2$ be defined as above, and $L=L_1\cup L_2$ be the set of affected data qudits. 
    
    When $\abs{L}\leq d-1$, we need $\abs{L_1}$ fresh ancillas, and at least $n-k-\dim(\mathcal{S}_{\mathcal{C}}^{\overline{L}})-\dim(\mathcal{S}_{\mathcal{C}}^{L_2 \setminus L_1}) - \dim (\mathcal{S}_{\mathcal{C}}^{L_1})-\dim (\mathcal{S}_\mathcal{C}^{L\overline{L}}\cap \mathcal{S}_{\mathcal{C}}^{\overline{L_1}}) + C$ stabilizer measurements to perform EEC, where 
    \[
    C=\begin{cases}
        1 & \text{ if }\exists s\in \mathcal{S}_\mathcal{C}^{L\overline{L}}\cap \mathcal{S}_{\mathcal{C}}^{\overline{L_1}} \text{ s.t. } [s,P^L]\neq 0,  \\
        0 & \text{otherwise}.
    \end{cases}
    \]

    For non-degenerate codes, 
    the minimum number of stabilizer elements to be measured is $n-k-\dim(\mathcal{S}_{\mathcal{C}}^{\overline{L}})$.
\end{lemma}

\begin{proof}[Proof Sketch]
    The proof is straightforward by Corollary~\ref{cor:min_msmt_code} and the cleaning lemma. 
    $\mathcal{S}_\mathcal{C}^{\overline{L}}$ and $\mathcal{S}_{\mathcal{C}}^{L_2\setminus L_1}$ are the unaffected local subgroups on unerased data qudits, $\mathcal{S}_\mathcal{C}^{L_1}$ denotes the local subgroup that can be prepared on the fresh ancillas, although this preparation is not realistic in practice. $\mathcal{S}_\mathcal{C}^{L\overline{L}}\cap \mathcal{S}_{\mathcal{C}}^{\overline{L_1}}$ can be fully unaffected or have exactly one generator affected by erasures, depending on the commutation relation between $P^L$ and this subgroup. 
    For non-degenerate codes, $\mathcal{S}_{\mathcal{C}}^{L_2\setminus L_1}$ and $\mathcal{S}_\mathcal{C}^{L_1}$ are trivial because any non-identity stabilizer has weight no smaller than $d$.
\end{proof}

\subsubsection{Shor-style SE circuit equipped with data TLDUs\label{sec:shor_dtldu}}

In Lemma~\ref{lem:shor_eec}, we have not considered how to detect data qudit losses yet. One way is to apply a TLDU to each data qudit $q_{a_j}$ \cite{perrin2025quantumerrorcorrectionresilient,baranes2025leveragingatomlosserrors}, which we call \emph{data TLDU} for convenience. There are a few concerns with this construction.

\paragraph{Extra overhead.} This modification doubles both the number of controlled operations between data and syndrome qudits and the number of ancillas in each SE circuit. 

\paragraph{Undetected losses.}\label{sec:undetected_loss_shor} If no losses occur during the stabilizer measurement, the syndrome information extracted from the stabilizer measurement circuit is valid; however, in Case $2$ and $3$ in Fig.~\ref{fig:3tldu_loss}, the loss of the fresh qudit is not detected and becomes an input loss in the next SE circuit \cite{perrin2025quantumerrorcorrectionresilient}, regardless of whether the data qudit gets lost in the same TLDU or not. 

In the current circuit, we only handle detected losses. Lemma~\ref{lem:shor_eec} then leads to the following corollaries. 

\begin{corollary}[Minimal EEC cost in Shor-style SE using data TLDUs]\label{cor:eec_cost_shor_tldu}
In a Shor-style SE circuit where data qudit losses are detected by subsequent data TLDUs, suppose that a subset $L_3\subseteq [n]$ of data qudits is lost during its corresponding TLDU. Then the minimal number of stabilizers to be measured for EEC is the same as in Lemma~\ref{lem:shor_eec}, except that $L$ is replaced by
\[
L := L_1\cup L_2\cup L_3 .
\]
Here, $L_1\cup L_3$ is the set of data qudits whose corresponding data TLDUs detect a qudit loss.

    When $\abs{L}\leq d-1$, at least $n-k-\dim(\mathcal{S}_{\mathcal{C}}^{\overline{L}})-\dim(\mathcal{S}_{\mathcal{C}}^{L_2 \setminus (L_1\cup L_3)}) - \dim (\mathcal{S}_{\mathcal{C}}^{L_1}) -\dim (\mathcal{S}_\mathcal{C}^{L\overline{L}}\cap \mathcal{S}_{\mathcal{C}}^{\overline{L_1\cup L_3}}) + C$ stabilizer measurements is required to perform EEC,  where 
    \[
    C=\begin{cases}
        1 & \text{ if }\exists s\in \mathcal{S}_\mathcal{C}^{L\overline{L}}\cap \mathcal{S}_{\mathcal{C}}^{\overline{L_1\cup L_3}} \text{ s.t. } [s,P^L]\neq 0,  \\
        0 & \text{otherwise}.
    \end{cases}
    \]

    For non-degenerate codes, 
    the minimum number of stabilizer elements to be measured is $n-k-\dim(\mathcal{S}_{\mathcal{C}}^{\overline{L}})$.
    
\end{corollary}

We also update the definition of affected data qubits.

\begin{definition}[Affected data qudits]
A data qudit is \emph{affected} by qudit losses in a Shor-style SE circuit equipped with data TLDUs, if at least one of this data qudit and its corresponding syndrome qudit is lost.
\end{definition}

\begin{corollary}[Validity of syndrome information in the presence of losses]
    Under the same settings and parameters as in Corollary~\ref{cor:eec_cost_shor_tldu}, if $L\neq \emptyset$, the classical measurement outcomes are useless and no useful syndrome information can be obtained; \footnote{Unless its local components, $P^L$ and $P^{\overline{L}}$, are also stabilizer elements; however, this case should be excluded in practice, since low-weight stabilizers are preferred for measurement.} if $L=\emptyset$, the measurement outcomes yield the actual syndrome information of $P$. 
\end{corollary}

Note that data qudit losses occurring during stabilizer measurement and those occurring during TLDUs are indistinguishable, and $L_3\cap L_1=\emptyset$ by definition. However, since both cases leave the data qudits \emph{affected}, the same EEC procedure applies. In such circuits, if any qudit loss is detected, i.e., if $L=L_1\cup L_2 \cup L_3 \neq \emptyset$, the syndrome information cannot be considered valid.

If the next stabilizer to be measured acts trivially on such data qudits, its syndrome information is not affected by the input losses; otherwise, the next measurement detects the loss, and we update the measurement sequence accordingly.

\paragraph{Losses detected during EEC.} In the entire syndrome measurement sequences, EEC can be viewed as just updating the measurement sequence, for the purpose of simultaneously converting erasures into located Pauli errors and refreshing syndrome information carried by the affected subgroup. Losses detected during EEC need to be converted as well. It will be handled in every single round of syndrome string extraction again by updating the measurement sequence using the new erasure information.

We defer a detailed description of how to update the measurement sequence to Section~\ref{sec:adaptive_protocol}.

\section{Corrections and conditions for the mixed error model\label{sec:correction}}

We now state FTEC conditions when both losses and Pauli errors are present. In the rest of this paper, we will focus on qubits, and the results easily extend to qudits of prime dimension. We do not further discuss qudits of composite dimension.

\subsection{Finding the correction operator
\label{sec:corr_operator}
}

From the adaptive erasure error correction introduced above, when each generator from the minimal set for erasure error correction is measured,  the syndrome information it carries is randomly drawn from $\{0,1\}$. 

We begin by defining a minimal sequence that converts an erasure into a located Pauli error.

\begin{definition}
    Suppose a set of erased qubits of state stabilized by code $\cC$ has been replaced with fresh ancillas. A sequence of generators is \emph{minimal}, if:
    \begin{itemize}
        \item By measuring this sequence, the state can be mapped back to the stabilizer subspace of $\cC$.
        \item By measuring any subsequence, the state cannot be mapped back to the stabilizer subspace of $\cC$. 
    \end{itemize}
\end{definition}

The syndrome information extracted from these measurements is random, as measuring each generator in a minimal sequence serves two purposes.

\begin{lemma}
    If an erasure is converted into a located Pauli error 
    by a minimal sequence of generator measurements, then the measurement of each generator produces a random bit of syndrome information.
\end{lemma}
\begin{proof}
    Suppose $G=\{g_1, \cdots, g_l\}$ is a minimal sequence of generators used for correcting erasures. For each $i$, the measurement of $g_i$ should effectively project back the state closer to the syndrome space. By the definition of measurement, if $g_i$ commutes with all stabilizer of the current state, $g_i$ does not change the state; otherwise, we can always choose a set of stabilizer generators so that $g_i$ anticommutes with exactly one of the generators, and the measurement of $g_i$ extracts a random output $\alpha_i$ 
    and replaces this generator with $\alpha_i g_i$ \cite[Section~10.5.3]{nielsen2010quantum}.

    If any $g_i$ commutes with all stabilizers of the state at the moment before its measurement, the subsequence $G\setminus \{g_i\}$ can be used to correct the same set of erasures as well, which leads to a contradiction. By minimality, each measurement of $g_i\in G$ should replace exactly one stabilizer generator of the state with $\alpha_i g_i$ and produce the random sign $\alpha_i$.
\end{proof}

Now, suppose we have erasures $L$ and an unlocated Pauli error $E$, along with a usable syndrome string. Existing loss-tolerant decoders provide efficient approaches to finding the correction operator, but the correctness of our protocol must be established independently of the decoder. It suffices to provide a lookup table that takes $L$ and a usable syndrome string as input, and outputs a correction operator. We now show how to construct such a correction operator.

\begin{theorem} 
    Given a state encoded by the stabilizer code $\cC$, a set of erased qubits $L$ with $\abs{L}\leq d-1$ that have been converted into a located Pauli error, and a usable syndrome string $\vec{s}_{use}$, 
    the correction operator for such a correctable combination of errors is uniquely determined by $(L, \vec{s}_{use})$, up to stabilizer operations.
\end{theorem}
\begin{proof}
    
Suppose erasures are converted into a random Pauli error $E_{\$}$ on $L$. Without loss of generality, we can assume that $E$ is a destabilizer, and $\operatorname{supp} E\subset \overline{L}$, since if the local component of $E$ on $L$ is nontrivial, it's overwritten by the randomness introduced by the located Pauli error $E_{\$}$.

Fix the set of generators to be the canonical generating set $S^L \sqcup S^{L\overline{L}}\sqcup S^{\overline{L}}$ based on the partition $\{L, \overline{L}\}$, and suppose
\[
    S^L = \{ a_i \otimes I_B\} , S^{\overline{L}} = \{ I_A \otimes b_j \}, 
    S^{L\overline{L}} = \{ g_k \otimes \overline{g}_k, h_k \otimes \overline{h}_k\}.
\]

W.l.o.g., we assume that $g_k, \overline{g}_k, h_k, \overline{h}_k$ are all destabilizers.

By assumption, the syndrome string $\vec{s}_{use}$ is exactly produced by the located error $E_{\$}$ and the unlocated Pauli error $E$

\[
\vec{s}_{use} = \sigma(E_{\$}\cdot E) = \sigma(E_\$)+\sigma(E).
\]

We can decompose the syndrome string as the concatenation of the substrings produced by generators of the three subgroups  
\[
\vec{s}_{use} = \vec{s}_{use}^L \circ \vec{s}_{use}^{L\overline{L}} \circ \vec{s}_{use}^{\overline{L}}.
\]

Let $\vec{s}_{\$} = \sigma(E_{\$})$, $\vec{s} = \sigma(E)$ and decompose them into substrings in the same way. We have 
\[
\vec{s}_{use}^{L\overline{L}}= \vec{s}_{\$}^{L\overline{L}}+ \vec{s}^{L\overline{L}}.
\]

If $\vec{s}_{\$}^{L\overline{L}}$ is known, we can correct $E_{\$}$ and $E$ separately. Let the correction operator for $E_\$$ be $E^c_\$$ and the one for $E$ be $E^c$. 
Let $I_G(\vec{s}), I_H(\vec{s})$ be the set of indices of elements of $S^{L\overline{L}}$ that correspond to a nontrivial bit in $\vec{s}$ and take the form $g\otimes \overline{g}$ and $h\otimes \overline{h}$ respectively.

Firstly notice that the generators in $S^L, S^{\overline{L}}$ are only affected by Pauli errors on $L, \overline{L}$, respectively. Since $E^c=E$ by the assumption that $E$ is reduced, we have 
\[
E^c = \mathcal{E}(\vec{s}) = E \sim  \mathcal{E}(\vec{0}^L\circ \vec{0}^{L\overline{L}}  \circ  \vec{s}_{use}^{\overline{L}}) \cdot  \prod_{k\in I_G(\vec{s}^{L\overline{L}})} h_k \cdot \prod_{k\in I_H(\vec{s}^{L\overline{L}})} g_k 
\]
and
\[
E_\$^c \sim \mathcal{E}(\vec{s}_{use}^L \circ \vec{0}^{L\overline{L}} \circ \vec{0}^{\overline{L}})\cdot \prod_{k\in I_G(\vec{s}_{\$}^{L\overline{L}})} h_k \cdot \prod_{k\in I_H(\vec{s}_{\$}^{L\overline{L}})} g_k.
\]

Their product 
\[
E^c_{tot} = E^c_{\$} \cdot E^c \sim \mathcal{E}(\vec{s}_{use}^L  \circ \vec{0}^{L\overline{L}} \circ \vec{s}_{use}^{\overline{L}}) \cdot  
\prod_{k\in I_G( \vec{s}_{use}^{L\overline{L}})} h_k \cdot \prod_{k\in I_H(\vec{s}_{use}^{L\overline{L}})} g_k
\]
By the cleaning lemma, 
any correction operator supported only on $L$ causes no logical error, so the RHS is a unique correction operator independent of $\vec{s}_{\$}^{L\overline{L}}$ that corrects this combination of errors. 
\end{proof}

\subsection{Strong and weak FTEC conditions}

The following definitions generalize the standard FTEC conditions to the mixed error setting, based on Theorem~\ref{thm:trade-off-mixed-correctability}.

\begin{definition}[Strong FTEC conditions under the mixed error model]
An error correction protocol is \emph{strongly $t$-fault} for an $[[n,k,d]]$ stabilizer code ($t\leq \lfloor\frac{d-1}{2}\rfloor$)
, if the following two conditions are satisfied:
\begin{enumerate}
    \item \textbf{Error correction correctness property (ECCP)}: If the input error is a combination of $e_{in}$ erased qubits and $p_{in}$ single-qubit Pauli errors with weight $r=\frac{1}{2}e_{in}+p_{in}$, and $e_{f}$ newly erased qubits and $p_{f}$ internal Pauli faults occur during the protocol with $s=\frac{1}{2}e_{f} + p_f, r+s\leq t$, then an ideal decoder outputs the same codeword whether it is applied to the input state or to the output state of the protocol.
    \item \textbf{Error-correction recovery property (ECRP):} Regardless of the weight of the input Pauli error, if $e_{f}$ qubits have been erased, and $p_{f}$ internal Pauli faults occur during the protocol with $s=\frac{1}{2}e_{f} + p_f \le t$, then the output state, with all erasures corrected, differs from a valid codeword by a Pauli error of weight at most $s$.

	\label{def:strong_FT_mixed}
\end{enumerate}
\end{definition}

\begin{definition}[Weak FTEC conditions under the mixed error model]
An error correction protocol is \emph{weakly $t$-fault} for an $[[n,k,d]]$ stabilizer code ($t\leq \lfloor\frac{d-1}{2}\rfloor$), 
if the following two conditions are satisfied:
	\begin{enumerate}
    \item \textbf{ECCP:} If the input error is a combination of $e_{in}$ erased qubits and $p_{in}$ single-qubit Pauli errors of weight $r=\frac{1}{2}e_{in}+p_{in}$, and $e_{f}$ newly erased qubits and $p_{f}$ internal Pauli faults occur during the protocol with $s=\frac{1}{2}e_{f} + p_f, r+s\leq t$, then an ideal decoder outputs the same codeword whether it is applied to the input state or to the output state of the protocol.
    \item \textbf{ECRP:} If the input error is a combination of $e_{in}$ erased qubits and $p_{in}$ single-qubit Pauli errors of weight $r=\frac{1}{2}e_{in}+p_{in}$, and $e_{f}$ newly erased qubits and $p_{f}$ internal Pauli faults occur during the protocol with $s=\frac{1}{2}e_{f} + p_f, r+s\leq t$, then the output state, with all erasures corrected, differs from a valid codeword by a Pauli error of weight at most $s$.
	\end{enumerate}
	
	\label{def:weak_FT_mixed}
\end{definition}

\section{Adaptive Shor-style syndrome measurements under the mixed error model\label{sec:adaptive_protocol}}

The process of converting erasures into located Pauli errors is simply measuring a sequence of stabilizers and can be integrated into the full syndrome measurement sequence. By dynamically selecting the next stabilizers to measure, we achieve adaptivity at both levels. In this section, we start by formally introducing the notations and building blocks. We then present syndrome measurement protocols that generalize those in \cite{tansuwannont2023adaptive}.

\subsection{One round of syndrome string extraction}
In a single round of syndrome string extraction, we want to obtain valid syndrome bits for a minimal generating set of stabilizers. This is the least requirement for finding a usable syndrome string to apply Pauli error correction.

Recall that in Section~\ref{sec:shor_dtldu}, we discussed that no valid syndrome bit can be extracted from a Shor-style SE circuit if any qubit loss is detected. A subgroup of stabilizers is affected by erasures, and its minimal generating set forms a minimal sequence for EEC.

We initially pick a set of generators. If we always want to pick the minimal sequence for EEC, the set of generators may change over time, requiring additional classical processing. However, the classical processing can be avoided by always using the same set of generators, and the increase in the number of measurements is generally acceptable for good codes.

Suppose during the stabilizer measurement of $g_i$, a set of data qubits $L$ is affected by erasures. We pick a canonical set of generators $G=S^L\sqcup S^{L\overline{L}}\sqcup S^{\overline{L}}$ based on the partition $L$ and $\overline{L}$. The obtained syndrome information can be remapped to the new syndromes for this set of generators.

The sequence of generators to be measured is then updated as the union of the set of generators in $G$ that are affected by erasure, and the set of generators in $G$ of which the syndrome information is not affected but has not been obtained yet. The length of the measurement sequence increases by at most $\abs{S^L} + \abs{S^{L\overline{L}}}$. If a different set of data qubits $L'$ is affected by an erasure at a later stage, we again update the sequence and syndrome information based on $L'$.

For any stabilizer code of distance $d$, the stopping and rejection conditions in a single syndrome extraction round are defined as follows:
\begin{itemize}
    \item \textbf{Stop when $d-1$ distinct data qubits have been erased and replaced.} 
In that case, due to the code's limited correctability, the protocol should stop and apply a correction if using the latest fully valid syndrome string, aiming only to correct the located error converted from the erasures. 
    \item \textbf{Reject when more than $d-1$ distinct data qubits have been erased and replaced.} Simply because this exceeds the correctability of the code.
\end{itemize}

The algorithm for one round of syndrome string extraction is given in Appendix~\ref{sec:subroutines}.

At the end of each round, we obtain a syndrome string $\vec{s}$ for a set of $r=n-k$ generators $G=\{g_1, \cdots, g_r\}$, as well as an erasure pattern about when and where the qubit losses are detected, and a syndrome string $\vec{s}$ of length $r$.

Let $p_j$ be the number of generators that have been measured before the stabilizer measurement, where the $j$-th data qubit is first erased in a round:
\begin{enumerate}
    \item If there is an input loss on qubit $j$, let $p_j = 0$;
    \item If qubit $j$ is firstly erased during measuring $g_k$, let $p_j=k$;
    \item If qubit $j$ is never erased in the current round, let $p_j=r+1$.
\end{enumerate}
Let $R=\{i: p_i \leq r\}$ be the set of qubits that have been erased at least once in this round. If $\vec{s}$ is fully valid, the usability of $\vec{s}$ and $R$ will be checked by Algorithm \ref{alg:tpb23_1}; otherwise, the entire syndrome measurement protocol is either stopped or rejected because the number of erasures reaches the limit.

\subsection{Syndrome extraction protocols satisfying the modified FTEC conditions\label{sec:strong_ft_new}}

Now let $G_i$ be the set of generators used at the end of round $i$, $R_i$ be the set of qubits that have been affected by erasure at least once in round $i$, $R_{> i} = \cup_{j> i} R_j$, and correspondingly, $R_{\leq i } = \cup_{j\leq i} R_j$. Let $\vec{s}_i$ be the syndrome string obtained in round $i$ using $G_i$ as the basis. Let $\varphi_i: \mathbb{F}_2^r \rightarrow \mathbb{F}_2^r$ be the map that transforms $\vec{s}_i$ from the basis $G_i$ to $G_{i+1}$, and $A_i=\{g: g\in G_i, \operatorname{supp}g\ \cap R_i \neq \emptyset  \}\subseteq G_i$ be the set of generators in $G_i$ whose syndrome information has been refreshed in round $i$. 

Given the erasure information and the syndrome string, if round $i$ has no internal Pauli faults, as long as the mixed error at the end of the round stays correctable, there exists a correction operator depending on $R_{\leq i}$ and $\vec{s}_i$ according to Section~\ref{sec:corr_operator}.

\paragraph{The difference vector.} Suppose that the input Pauli error of round $i$ is $E_i$. If no internal Pauli fault occur during round $i$ and round $i+1$, the two consecutive syndrome string $\vec{s}_i$ and $\vec{s}_{i+1}$ satisfy the following relation 
\begin{equation}\label{eq:syn_str_comp}
\varphi_i(\vec{s}_i)|_{\overline{A_{i+1}}} = \vec{s}_{i+1}|_{\overline{A_{i+1}}}.
\end{equation}

We now track whether we can detect an internal Pauli fault in every two consecutive rounds. We adopt the term \emph{difference vector} from \cite{tansuwannont2023adaptive} and borrow some related notions. The difference between two consecutive syndrome strings in the Pauli-error model directly implies the existence of internal faults.  Note that the operational role of the difference vector is as follows:
\begin{enumerate}
    \item If two consecutive syndrome strings differ, we can ensure that at least one Pauli fault occurs during round $i$ and round $i+1$.
    \item If two consecutive syndrome strings are identical, we cannot distinguish between the case where no Pauli fault occurs and the case where the syndrome patterns produced by multiple faults cancel.
\end{enumerate}

Under the mixed error model, we instead check syndrome substrings. Denote the difference vector by $\vec{\delta}$. Let $\delta_i = 0$ if \eqref{eq:syn_str_comp} is satisfied; otherwise, let $\delta_i=1$.

Recall that in \cite{tansuwannont2023adaptive}, $\vec{\delta}$ is viewed as a sequence of all-zero substrings split by ones, $\vec{\delta}= \eta_11 \eta_2 1 \cdots 1\eta_c$, where the first and the last all-zero substrings can have zero length. We may compute $\alpha_j$ (the minimum possible number of Pauli faults before the $j$-th all-zero vector), $\beta_j$ (the minimum possible number of Pauli faults after the $j$-th all-zero vector), $\gamma_j$ (the maximum possible number of Pauli faults that can occur in the rounds corresponding to the $j$-th all-zero vector), and then find a usable all-zero substring in the same way based on $\vec{\delta}$.

\begin{algorithm*}[ht]\caption{(Restatement of Algorithm 1 in \cite{tansuwannont2023adaptive}) finds a usable all-zero substring}\label{alg:tpb23_1}
\AlgoDontDisplayBlockMarkers\SetAlgoNoEnd\SetAlgoNoLine%
\SetKwInOut{KwIn}{Input}
\SetKwInOut{KwOut}{Output}
\SetKw{KwTo}{to}
\SetKw{KwDownTo}{downto}
\SetKwFunction{IsErasureUncorrectable}{IsErasureUncorrectable}
\KwIn{
$t_\text{in}$, the maximal number of correctable Pauli faults.
$\delta_\text{in} = \eta_1 1\eta_21\cdots 1\eta_c$, a difference vector with $c$ all-zero substrings
}
\KwOut{A usable all-zero substring $\eta_j$, or a symbol $\times$ that represents none of the substrings $\eta_j$ is usable.}
\For{$j\gets c$ \KwDownTo $1$}{
    calculate $\alpha_j, \beta_j$

    $\gamma_j \gets \abs{\eta_j}$
    
    \If{$\alpha_j+\beta_j+\gamma_j\geq t_\text{in}$}{
        \Return{$\eta_j$}
    }
}
\Return{$\times$}
\end{algorithm*}

\paragraph{The time overhead.} 
An FTEC protocol correcting up to a combination of internal erasures and Pauli faults with weight at most $t$ can take the parameter $t_{in}=\lfloor t-e/2  \rfloor$ to be less than $t$ when erasures occur, where $e$ is the number of data qubits that have been affected by erasures. Therefore, the presence of erasure may let the protocol stop and apply the correction at an earlier stage. However, the worst case occurs when only Pauli errors are present. Therefore, the maximum number of rounds remains the same. Since the protocol can tolerate up to $d-1$ erasures, there are at most $2(d-1)$ additional stabilizer measurements in the entire sequence. The rejection condition is also defined as ``reject when more than $d-1$ distinct data qubits have been affected by erasures''.

\paragraph{The correction operator.} 
The usable syndrome string may be from the middle of multiple rounds. We only need to apply the correction to the subset of qubits that are not later affected by erasures. Suppose the total number of syndrome string extraction rounds is $N$. 
For the usable syndrome string $s$ obtained from $\eta_j = \delta_{s_j} \cdots \delta_{t_j}$, the corresponding correction operator $E^c_i$ according to Section~\ref{sec:corr_operator} is able to correct a combination of located and unlocated Pauli errors at the end of that round.

The Pauli correction we apply at the end is $E^c_i |_{\overline{R_{>i}}}$. 
Note that there may be multiple rounds of syndrome string extraction yielding a usable all-zero substring; however, the resulting Pauli correction operator is the same in all such cases.

Protocols satisfying the strong and weak FTEC conditions under the mixed error model can be extended from the original protocols under the standard Pauli error model \cite{tansuwannont2023adaptive}, differing in the rejection condition, the computation of the difference vector $\delta$, and the choice of the Pauli correction operator. 

We now present the protocol that satisfies the modified strong FTEC conditions.

\begin{protocol}{loss-tolerant FTEC protocol satisfying the strong FTEC conditions for a stabilizer code of arbitrary distance}
\label{pro:strong_dany}
    Given a generating set $S$ of the code, an input state $\rho$, perform the following:
    \begin{enumerate}
        \item Initially set $R=\emptyset$.
        \item In the $i$-th round, obtain $(\vec{s}_i, R_i, \rho_i, S_i)$ from \textsc{ SyndromeStringExtraction}$(\rho, S, R)$. Update $R$ by $R\gets R_i$. If a rejection symbol is returned, 
        \textbf{reject}; if $\vec{s}_i$ is not fully valid, \textbf{stop the protocol and look for a correction}.  
        \item 
        
        After the $i$-th round ($i\geq 2$) compute $\delta_{i-1}$. \textbf{Stop the protocol and look for a correction} if one of the following conditions is satisfied:
        \begin{enumerate}
            \item If at least one usable $\eta_j$ is found by Algorithm $1$, with $t_{in} =\lfloor  t-\abs{R_i}/2\rfloor, \delta_{in} = \delta_1\cdots \delta_{i-1}$, stop the protocol and perform Pauli error correction based on the erasure information and the syndrome string from a round covered by $\eta_j$.
            \item If the total number of non-overlapping $11$ substrings in $\vec{\delta}$ is exactly $\lfloor t-\abs{R_i}/2 \rfloor$, stop and perform Pauli error correction using the correction operator $E'$ corresponding to the latest fully valid syndrome string.
        \end{enumerate}        
        \item 
        Repeat syndrome measurements unless the protocol is either rejected or stopped.  If the protocol stops but no fully valid and usable syndrome string exists, \textbf{reject}.
    \end{enumerate}
\end{protocol}

The FTEC conditions are satisfied because the Pauli correction we apply removes the destabilizer components from the input error and from Pauli faults before a certain round. Pauli faults that occur or are introduced by erasure error correction after that round remain in the correctable region, as long as the total number of internal erasures and internal Pauli faults does not exceed the limit.

In very rare cases (when we learn that the combination of errors must have exceeded the code's correctability), the protocol is rejected; otherwise, the upper bound on the number of syndrome string extraction rounds depends on both the code distance and the number of erased qubits.

\begin{remark}
    If $e$ data qubits have been erased, the maximum rounds of syndrome string extraction is
    \[
    \begin{cases}
        \max(1, (\frac{t'+3}{2})^2-3) & t' \text{ odd},\\
        \max((\frac{t'+2}{2})(\frac{t'+4}{2})-3) & t'\text{ even} 
    \end{cases}
    \]
    with $t' = \lfloor t-\frac{\#\ erasures}{2} \rfloor$ and $t=\lfloor \frac{d-1}{2}\rfloor$.
\end{remark}

The protocol satisfying the weak FTEC conditions can be generalized in a similar manner; we omit the details here for brevity.

\section{Discussion and future directions \label{sec:discuss}}

There are several results we have not had the time to treat in this version: 
\begin{enumerate}
    \item The numerical simulation results. Since our adaptive protocols require immediate decision-making, we face the same issues as \cite{berthusen2025adaptive}, that is, we cannot take advantage of the high-performance features provided by current packages for stabilizer simulation, but we have to instead simulate the protocols step by step. Therefore, we leave them to the next version.
    \item Results about the minimal overhead of erasure error correction on SE gadgets of other styles. Roughly speaking, regardless of whether the syndrome is obtained by one bit at a time or multiple syndrome bits are obtained at once, the overhead depends on the number of data qubits affected by erasures, while the validity of syndrome information is affected by any type of qubits that are erased. We then conjecture that FTEC schemes with more ancillary qubits per circuit are less loss-tolerant\footnote{Qubits that can be measured and reset are counted once per reset operation.}. No background knowledge beyond the stabilizer formalism is required to obtain these results, so we have decided to delegate all the tedious details to a future version or to a separate article.
    \item Results for other loss detection methods. So far we have only discussed the use of data TLDUs in this paper. For other loss detection methods, the syndrome information may be affected differently. For example, in SWAP-based syndrome extraction, the validity of a measured syndrome bit may not be known immediately; therefore the syndrome measurement protocols need to be modified accordingly.
\end{enumerate}

We outline possible directions for future work.   
\paragraph{Noise models.}
So far we have only considered the standard Pauli error model together with qubit losses. 
However, more realistic noise models may include additional effects, such as imperfect teleportation and biased noise. 
It would also be interesting to extend the analysis to adversarial noise models and to settings where qubit loss occurs in interactive protocols.
\paragraph{The mixed error correctability.} Several papers \cite{baspin2022connectivity, dai2025locality} have discussed the trade-off between the size of the correctable region and the connectivity of a code. It may be worth discussing the correctable set of unlocated Pauli errors, or the maximum possible weight of such correctable Pauli errors 
given a set of erasures. 

\paragraph{Code puncturing.} A stabilizer code can be punctured by deleting several coordinates (physical qudits). Recall that in Section~\ref{sec:loss_erasures_and_loss_detection}, we mentioned that regardless of how a data qubit is affected by qubit losses, the affected stabilizer subgroup will stay the same. This holds when multiple qubits are lost. The size of the minimal generating set of the affected subgroup upper bounds the additional overhead introduced by adaptive erasure error correction, namely the increase in the measurement sequence length. This is, in turn, loosely upper bounded by twice the number of lost qubits, a number independent of the code parameters. We think this is conceptually related to Open Problem 2 of \cite{delfosse2020short}, which asks whether puncturing one coordinate increases the length of a measurement sequence by at most an additive constant. It is also possible to use results from puncturing \cite{grassl2021algebraic, gundersen2025puncturing}, though the intuition of puncturing codes is more likely to provide a way to construct good codes from existing ones, different from ours.

\paragraph{Overhead.} The FT overhead is also related to the stabilizer weights. In our protocol, we choose the canonical generating set of stabilizer groups according to the bipartite partition of the system. The total weight of these generators is not required to be minimal, and updating these generators in this way may not always be practical. We have not taken this next step, as several prior work has not either. This is not a major issue for good codes with low-weight stabilizers. We hope that someone will do this in the near future.
\paragraph{Code concatenation. }\cite{berthusen2025adaptive} constructs a family of codes by concatenating a small detection code with a hypergraph code. Using low-level error detection, their scheme locates Pauli errors and thereby reduces the number of stabilizers to be measured. The underlying idea is similar to adaptive erasure error correction. It would be interesting to investigate whether code concatenation can also improve loss tolerance, though it is a different problem from the one studied here.
\section{Acknowledgments}
YW would like to thank Prof.~Todd Brun for guidance and support above and beyond the call of duty; the Departments of Electrical and Computer Engineering and Mathematics for their support; Zihan Xia for many helpful discussions over the years, including early-stage brainstorming on qudit codes dating back to the summer of 2022, long before this work took shape, conversations about adaptive erasure error correction during our trip to the UCLA East Asian Library on Oct.~9,~2025, near the completion of this part, our occasional reading and discussions of \cite{delfosse2020short, tansuwannont2023adaptive}, and for following the ups and downs of this work; Miles Gorman for an insightful discussion on syndrome extraction, which saved me from getting lost in too many schemes; and last but not least, Edward Kim for many recent related conversations in our office, for our shared research interests and experiences, though we missed the chance to work together. 

\medskip
\begin{CJK*}{UTF8}{bsmi}
\noindent\textit{敘懷經久不成悲，許將今日暫伸眉。}
\end{CJK*}
\printbibliography
\appendix
\section{Proofs\label{sec:proofs}}
\subsection{Proof of Theorem~\ref{thm:min_msmt_general}\label{proof:min_msmt_general}}
\begin{proof}
    The stabilizer tableau of the final state on $\cH_A\otimes \cH_B$, up to phases on the generators, can be written as
    \[
        \begin{array}{c}
    \overbracket{\phantom{sssss}}^{B} \phantom{ss}\overbracket{\phantom{sssss}}^{A}\\
    \begin{array}{ccc}
      I &  & \mathcal{S}'^A \\
      &\mathcal{S}'^{BA}  &\\
      \mathcal{S}'^B & &I
    \end{array}
    \end{array}
    \]
     We can prepare the $\abs{B}$ ancillary qub(d)its on $\cH_B$ optimally. Suppose the state is stabilized by $\mathcal{S}_{anc}$. We may choose $\mathcal{S}_{anc}= \mathcal{S}'^B$ in the special (and trivial) case $\abs{B}= \dim \mathcal{S}'^B$, and $\mathcal{S}_{anc}\supset \mathcal{S}'^B$ otherwise. 
    The stabilizer tableau of the initial joint state on $\cH_B\otimes \cH_A\otimes \cH_E$ can be written as:
    \[
    \begin{array}{l}
    \phantom{s}\overbracket{\phantom{ss}}^{B}\overbracket{\phantom{ssssssss}}^{A} \overbracket{\phantom{ssssssss}}^{E}\\
    \begin{array}{ccccl}
      I & \mathcal{S}_1 &  & I &\text{no need to convert}\\
      I & \overline{\mathcal{S}_1} &  & I &\text{\text{}}\\
      I & &\mathcal{S}^{AE}  &&\text{}\\
      I & I & &\mathcal{S}^E&\text{to be discarded}\\
      \mathcal{S}'^B & & I && \text{no need to convert}\\
      \overline{{\mathcal{S}'}^B}
      & & I && \text{}
    \end{array}
    \end{array}
    \]
    where 
    $\mathcal{S}^{AE} = \mathcal{S}\setminus (\mathcal{S}^A\sqcup \mathcal{S}^E)$, 
    $\mathcal{S}_1 = \mathcal{S}^A\cap K\mathcal{S}'^A$,  
    $\overline{\mathcal{S}_1} = \mathcal{S}^A\setminus \mathcal{S}_1$, 
    and $\overline{{\mathcal{S}'}^B} = \mathcal{S}_{anc}\setminus \mathcal{S}^B$. If any of them is empty, the corresponding rows are then removed from the tableau.
    Note that $\mathcal{S}^A$ is decomposed into $\mathcal{S}_1\sqcup \overline{\mathcal{S}_1}$, and 
    the extension of $\mathcal{S}_1$ 
    on $\cH_A\otimes \cH_B$ is contained in $K\mathcal{S}'$.
    
    It follows that to convert this state to the final state on $\cH_A\otimes \cH_B$, the number of stabilizers required to be measured is equal to the dimension of the subspace of $\mathcal{S}$ which requires projection:
        \[  \dim \mathcal{S}' - \dim \mathcal{S}_1-\dim \mathcal{S}^B= 
    \dim \mathcal{S}' - \dim ((\mathcal{S}^A/K) \cap (\mathcal{S}'^{A}/K))
    -\dim \mathcal{S}^B.
    \]
    
\end{proof}
\subsection{Proof of Corollary~\ref{cor:min_msmt_code}\label{proof:min_msmt_code}}
\begin{proof}
    Let $\mathcal{S}_L$ be the stabilizer subgroup carrying the unknown logical information on $\cH_A\otimes \cH_E$ and $\mathcal{S}_{\mathcal{C}L}$ be the subgroup carrying the same logical information on $\cH_A\otimes \cH_B$. 

    By the cleaning lemma, when $E$ is correctable, there exists a set of generators supported only on $A$ that generates both $\mathcal{S}_L$ and $\mathcal{S}_{\mathcal{C}L}$. Obviously, we have
    \[
    \mathcal{S}_{\mathcal{C}L}/K \cong \mathcal{S}_L /K, \dim \mathcal{S}_L = \dim \mathcal{S}_L^A = \dim \mathcal{S}_{\mathcal{C}L}.
    \]

    The stabilizer group of the initial state is
    $\mathcal{S} \sqcup \mathcal{S}_L$ and the stabilizer group of the resulting state is
    $\mathcal{S}_\cC \sqcup \mathcal{S}_{\mathcal{C}L}$.
    By Theorem~\ref{thm:min_msmt_general}, the minimum number of measurements for converting the state is 
    \begin{align*}
&\dim(\mathcal{S}_\cC \sqcup \mathcal{S}_{\mathcal{C}L})  
- \dim\!\big(((\mathcal{S}_\cC \sqcup \mathcal{S}_{\mathcal{C}L})^A/K) \cap ((\mathcal{S} \sqcup \mathcal{S}_L)^A/K)\big)  \\
&\quad - \dim \mathcal{S}_\cC^{B}    \\
&= \dim \mathcal{S}_\cC + \dim \mathcal{S}_L^A 
- \dim \big((\mathcal{S}_\cC^A /K) \cap (\mathcal{S}^A/K)\big) \\
&\quad - \dim \mathcal{S}_{\mathcal{C}L} 
- \dim \mathcal{S}_\cC^{B} \\
&= \dim \mathcal{S}_\cC 
- \dim \big((\mathcal{S}_\cC^A /K) \cap (\mathcal{S}^A/K)\big) 
- \dim \mathcal{S}_\cC^{B}.
\end{align*}

    For nondegenerate codes, the weights of all physical representations of all non-identity stabilizers are not smaller than the distance. Given the constraints $\abs{A}\geq n-d+1$ and $\abs{B}<d$, for nondegenerate codes we have $\dim \mathcal{S}_\cC^B = 0$.
\end{proof}

\section{Subroutines\label{sec:subroutines}}

\subsection{Notations and building blocks}
We define several basic operations before presenting the protocols in the following sections.
\begin{definition}
    Let $\mathcal{C}$ be the stabilizer code, assumed implicitly in all subroutines.

    \textsc{Measure$(\rho, g)$} denotes a noisy stabilizer measurement of $g$ on the state $\rho$, returning $(s, L, R, \rho')$, where 
    \begin{itemize}
        \item $s\in \{0,1,*, \times\}$ is the measurement outcome, either a valid syndrome bit, an invalid syndrome bit $*$ produced by qubit loss, or a rejection symbol $\times$ indicating too many erasures.
        \item $L$ is the set of data qubits affected by qubit loss during the measurement.
        \item $R$ is the set of data qubits erased during the measurement.\footnote{The undetected loss is treated as an input error of the next stabilizer measurement.}
        \item $\rho'$ is the post-measurement state.
    \end{itemize}
    \textsc{MinimalGeneratingSet($S$)} returns a minimal generating set of a set of stabilizer elements $S$.

    \textsc{PauliCorrection($\rho, s, S^R$)} denotes the  Pauli recovery applied to the state $\rho$, based on the syndrome string $s$ obtained from measuring the stabilizers $S^R$.

    \textsc{IsErasureUncorrectable}$(R)$ returns whether the set of erased qubits $R$ is uncorrectable.
    
    \textsc{CanonicalGeneratingSet}$(\mathcal{S}, A)$ returns a minimal generating set of $\mathcal{S}$ with respect to the partition $(A,\overline{A})$.
    
    \textsc{AffectedGenerators}$(\mathcal{S}, A)$ returns a minimal generating set of $\mathcal{S}^A\sqcup \mathcal{S}^{A\overline{A}}$ that is guaranteed to be a subset of \textsc{CanonicalGeneratingSet}$(\mathcal{S}, A)$.

    \textsc{ChangeBasis}$(syn, S, S')$ takes a dictionary of syndrome information $syn$ and updates it according to the change of generator set from $S$ to $S'$. The syndrome information for $g\in S'$ is valid if and only if there exists a subset $\{g_i\}\subseteq S$ such that $\prod_i g_i = g$ and the syndrome bit for every $g_i$ is valid.

    \textsc{SyndromeString}$(syn, S)$ returns the syndrome string implicitly ordered by the generators in $S$.
 
\end{definition}

For simplicity's sake, we assume that the code is non-degenerate, so that the minimal EEC sequence depends only on $n,k$ and the local stabilizer subgroup of affected data qubits. For degenerate codes, more detailed erasure information needs to be returned by \textsc{Measure}$(\rho, g)$. In \textsc{Measure}$(\rho, g)$, erased data qubits are replaced by fresh ancillas via data TLDUs.

\newpage
\subsection{One round of syndrome string extraction}

\begin{algorithm}[ht]\caption{{\sc SyndromeStringExtraction$(\rho, S, R_0)$} performs one round of syndrome string extraction}
\label{alg:full_syn_ext}
\AlgoDontDisplayBlockMarkers\SetAlgoNoEnd\SetAlgoNoLine%
\SetKwInOut{KwIn}{Input}
\SetKwInOut{KwOut}{Output}
\SetKwFunction{IsErasureUncorrectable}{IsErasureUncorrectable}
\SetKwFunction{SyndromeBitExtraction}{SyndromeBitExtraction}
\SetKwFunction{AffectedGenerators}{AffectedGenerators}
\SetKwFunction{CanonicalGeneratingSet}{CanonicalGeneratingSet}
\SetKwFunction{Measure}{Measure}
\SetKwFunction{ChangeBasis}{ChangeBasis}
\SetKwFunction{SyndromeString}{SyndromeString}
\KwIn{An unknown code state $\rho$.\\
$S=\{g_1, g_2, \cdots, g_{r}\}$, the minimal generating set of $\mathcal{S}_{\mathcal{C}}$ chosen at the beginning of the round.\\
$R_0$, A set of data qubits that have been erased and replaced before this round.\\
}
\KwOut{A tuple $(\vec{s},  R, \rho, S')$.\\
$\vec{s} \in (\{0,1,*\}^r) \cup \{\times\}$ represents a syndrome bitstring, or a rejection symbol $\times$ indicating that the algorithm is stopped due to too many erasures. If $\vec{s}$ contains any $*$, then $\vec{s}$ is not fully valid.\\
$R$ is the set of data qubits that have been erased between the beginning of the entire protocol and the end of this round.\\
$\rho$ is the post-measurement state.\\
$S'$ is the minimal generating set of $\mathcal{S}_{\mathcal{C}}$ chosen at the end of the round.
}

$syn\gets \{\}$

\For{$g\in S$}{
    $syn[g] \gets *$
    
    \tcp{$*$ denotes that a valid syndrome bit for $g$ has not yet been obtained}
}

$R\gets R_0$

$M \gets S$ \tcp{the set 
of generators to be measured
}

\While{$M\neq \emptyset$}{
    $g\gets M.pop()$ \tcp{pop the first element from $M$}

    $(s_g,L_g,R_g, \rho)\gets \Measure(\rho, g)$

    $R\gets R\cup R_g$
  
    \If{$s_g=\times$ \textbf{or} $\abs{R}\geq d$}{
        \Return{$(\times, R,  \rho,\{\})$}
    }
    \If{$\abs{R}=d-1$}
    {
        $\vec{s} \gets \SyndromeString(syn, S)$
        
        \Return{$(\vec{s}, R, \rho, S)$}
    }
    
    \uIf{$s_g\neq *$}{
        $syn[g] \gets s_g$
    }
    \Else{
        $A_g\gets \AffectedGenerators(\mathcal{S}_{\mathcal{C}}, L_g)$
        
        $S_g\gets \CanonicalGeneratingSet(\mathcal{S}_{\mathcal{C}}, L_g)$

        $\ChangeBasis(syn, S, S_g)$

        $M\gets A_g \cup \{g\in S_g :syn[g]=*\} $

        $S\gets S_g$

    }
    
}
$\vec{s} \gets \SyndromeString(syn, S)$

\Return{$(\vec{s}, R,\rho, S)$}
\end{algorithm}

\end{document}